%% file: main-hyper-local-ratio.tex
\documentclass[sigconf]{acmart}
\setcopyright{none}

\usepackage[utf8]{inputenc}
\usepackage{graphicx}
\usepackage{latexsym}
\usepackage{amsmath}
\usepackage{amssymb}
\usepackage{url}
\usepackage{float}
\usepackage{tikz}
\usepackage{fancyvrb}
\usepackage{listings}
\usepackage{xcolor}
\usepackage{color}
\usepackage{soul}
\usepackage{multirow}
\usepackage{multicol}
\usepackage{subfig}
\usepackage{booktabs}
\usepackage{caption}

\usepackage{paralist}

\copyrightyear{}
\acmYear{}
\setcopyright{none}
\acmConference{}
\acmBooktitle{}
\acmPrice{}
\acmDOI{}
\acmISBN{}

\usepackage{xcolor}

\newcommand{\deltasplit}{$\delta$-linear threshold}

\newcommand{\xhdr}[1]{\vspace{0.5mm}\noindent{\textbf{#1.}}\hspace{0.5mm}}

\PassOptionsToPackage{hyphens}{url}\usepackage{hyperref} 
\definecolor{mylinkcolor}{RGB}{0,0,140}
\hypersetup{colorlinks,allcolors=mylinkcolor,citecolor=mylinkcolor}
\usepackage[capitalize]{cleveref}

\input{preamble}

\begin{document}

\title{Minimizing Localized Ratio Cut Objectives in Hypergraphs}


%

\begin{abstract}
	Hypergraphs are a useful abstraction for modeling multiway relationships in data, 
and hypergraph clustering is the task of detecting groups of closely related nodes in such data.
\emph{Graph} clustering has been studied extensively, and there are numerous methods for detecting small, localized clusters without having to explore an entire input graph. However, there are only a few specialized approaches for localized clustering in hypergraphs. Here we present a framework for local hypergraph clustering based on minimizing localized ratio cut objectives. Our framework takes an input set of reference nodes in a hypergraph and solves a sequence of hypergraph minimum $s$-$t$ cut problems in order to identify a nearby well-connected cluster of nodes that overlaps substantially with the input set. 

Our methods extend graph-based techniques but are significantly more general and have new output quality guarantees. First, our methods can minimize new generalized notions of hypergraph cuts, which depend on specific configurations of nodes within each hyperedge, rather than just on the number of cut hyperedges. Second, our framework has several attractive theoretical properties in terms of output cluster quality. Most importantly, our algorithm is strongly-local, meaning that its runtime depends only on the size of the input set, and does not need to explore the entire hypergraph to find good local clusters. 
We use our methodology to effectively identify clusters in hypergraphs of real-world data with millions of nodes, millions of hyperedges, and large average hyperedge size with runtimes ranging between a few seconds and a few minutes.
\end{abstract}

\author{Nate Veldt}
\affiliation{
	\institution{Cornell University}
	\department{Center for Applied Mathematics}
}
\email{nveldt@cornell.edu}

\author{Austin R.~Benson}
\affiliation{%
	\institution{Cornell University}
	\department{Department of Computer Science}
}
\email{arb@cs.cornell.edu}

\author{Jon Kleinberg}
\affiliation{%
	\institution{Cornell University}
	\department{Department of Computer Science}
}
\email{kleinberg@cornell.edu}

\maketitle

\input{v1-intro}
\input{v3-prelims}

\input{v2-hypergraph-objectives}

\input{v3-strong-locality}

\input{v1-cut-improvement}

\input{v2-experiments}

\input{v1-conclusion}

\begin{acks}
We acknowledge and thank several funding agencies. This research was supported by a Vannevar Bush Faculty Fellowship, a Simons Investigator grant, NSF Award DMS-1830274, ARO Award W911NF19-1-0057, and ARO MURI.
\end{acks}

\bibliographystyle{plain}
\bibliography{short}
\include{sec-appendix-proofs}

\end{document}

%% file: preamble.tex
\usepackage{mathtools}
\usepackage{algorithm}
\usepackage{algorithmicx}
\usepackage{algpseudocode}
\usepackage{subfig}

\algdef{SE}[DOWHILE]{Do}{doWhile}{\algorithmicdo}[1]{\algorithmicwhile\ #1}%

       \newcommand{\vlm}{\textbf{v}}
	\newcommand{\vbr}{\vol(\bar{R})}
\newcommand{\hlc}{\textbf{HLC}}
\newcommand{\cond}{\textbf{cond}}
\newcommand{\lcond}{\textbf{local-cond}}

\newcommand{\ncut}{\textbf{ncut}}
\newcommand{\cut}{\textbf{cut}}

\newcommand{\vol}{\textbf{vol}}
\newcommand{\volh}{\textbf{vol}_\mathcal{H}}

\renewcommand{\vec}[1]{\boldsymbol{\mathrm{#1}}}

\DeclareMathOperator*{\minimize}{minimize}

\DeclareMathOperator{\argmin}{argmin}

\providecommand{\vw}{\ensuremath{\vec{w}}}

%% file: v1-intro.tex

\section{Introduction}
Graphs are a common mathematical abstraction for modeling pairwise interactions between objects in a dataset.
A standard task in graph-based data analysis is to identify well-connected \emph{clusters} of nodes, which share more edges with each other than the rest of the graph~\cite{schaeffer2007graphclustering}.
For example, detecting clusters in a graph is used to identify communities~\cite{fortunato2016communitydetection}, predict class labels in machine learning applications~\cite{blum2001learning}, and segment images~\cite{shimalik-ncut}.
Standard models for such clusters are \emph{ratio cut} objectives, which measure the ratio between the number of edges leaving a cluster (the \emph{cut}) and some notion of the cluster's size (e.g., the number of edges or nodes in the cluster); common ratio cut objectives include conductance, sparsest cut, and normalized cut.
Ratio cut objectives are intimately related to spectral clustering techniques, with the latter providing approximation guarantees for ratio cut objectives (such as conductance) via so-called \emph{Cheeger} inequalities~\cite{chung1997spectral}.
In some cases, these ratio cut objectives are optimized over an entire graph to solve a global clustering or classification task~\cite{shimalik-ncut}.
In other situations, the goal is to find sets of nodes that have a small ratio cut and are localized to a certain region of a large graph~\cite{AndersenChungLang2006,veldt16simple,Orecchia:2014:FAL:2634074.2634168,Andersen:2008:AIG:1347082.1347154}.

Recently, there has been a surge of hypergraph methods for machine learning and data mining~\cite{Agarwal2005beyond,Agarwal2006holearning,Zhou2006learning,panli2017inhomogeneous,panli_submodular,benson2018simplicial,chitra2019random,yadati2019hypergcn}, as hypergraphs can better model multiway relationships in data.
Common examples of multiway relationships include academic researchers co-authoring papers, retail items co-purchased by shoppers, and sets of products or services reviewed by the same person. 
Due to broader modeling ability, there are many hypergraph generalizations of graph-cut objectives,
including hypergraph variants of ratio cut objectives like conductance and normalized cut~\cite{BensonGleichLeskovec2016,Zhou2006learning,panli2017inhomogeneous,chan2018spectral,chan2018generalizing}.

Nevertheless, there are numerous challenges in extending graph-cut techniques to the hypergraph setting, and current methods for hypergraph-based learning are much less developed than their graph-based counterparts.
One major challenge in generalizing graph cut methods is that the concept of a cut hyperedge --- how to define it and how to penalize it --- is more nuanced than the concept of a cut edge.
While there is only one way to separate the endpoints of an edge into two clusters, there are several ways to split up a set of three or more nodes in a hypergedge.
Many objective functions model cuts with an \emph{all-or-nothing} penalty function, which assigns the same penalty to any way of splitting up the nodes of the hyperedge (and a penalty of zero if all nodes in the hyperedge are placed together)~\cite{lawler1973,ihler1993modeling,hadley1995} (\cref{fig:aon}).
However, a common practical heuristic is a \emph{clique expansion}, which replaces each hyperedge with a weighted clique in a graph~\cite{panli2017inhomogeneous,hadley1995,Zhou2006learning,zien1999,BensonGleichLeskovec2016} (\cref{fig:clique}).
The advantage is that graph methods can be directly applied, but this heuristic actually penalizes cut hyperedges differently than the all-or-nothing model.
Another downside is that for hypergraph with large hyperedges, clique expansion produces a very dense graph.
\begin{figure}[t]
	\centering
	\subfloat[All-or-nothing cut \label{fig:aon}] 
	{\includegraphics[width=.33\linewidth]{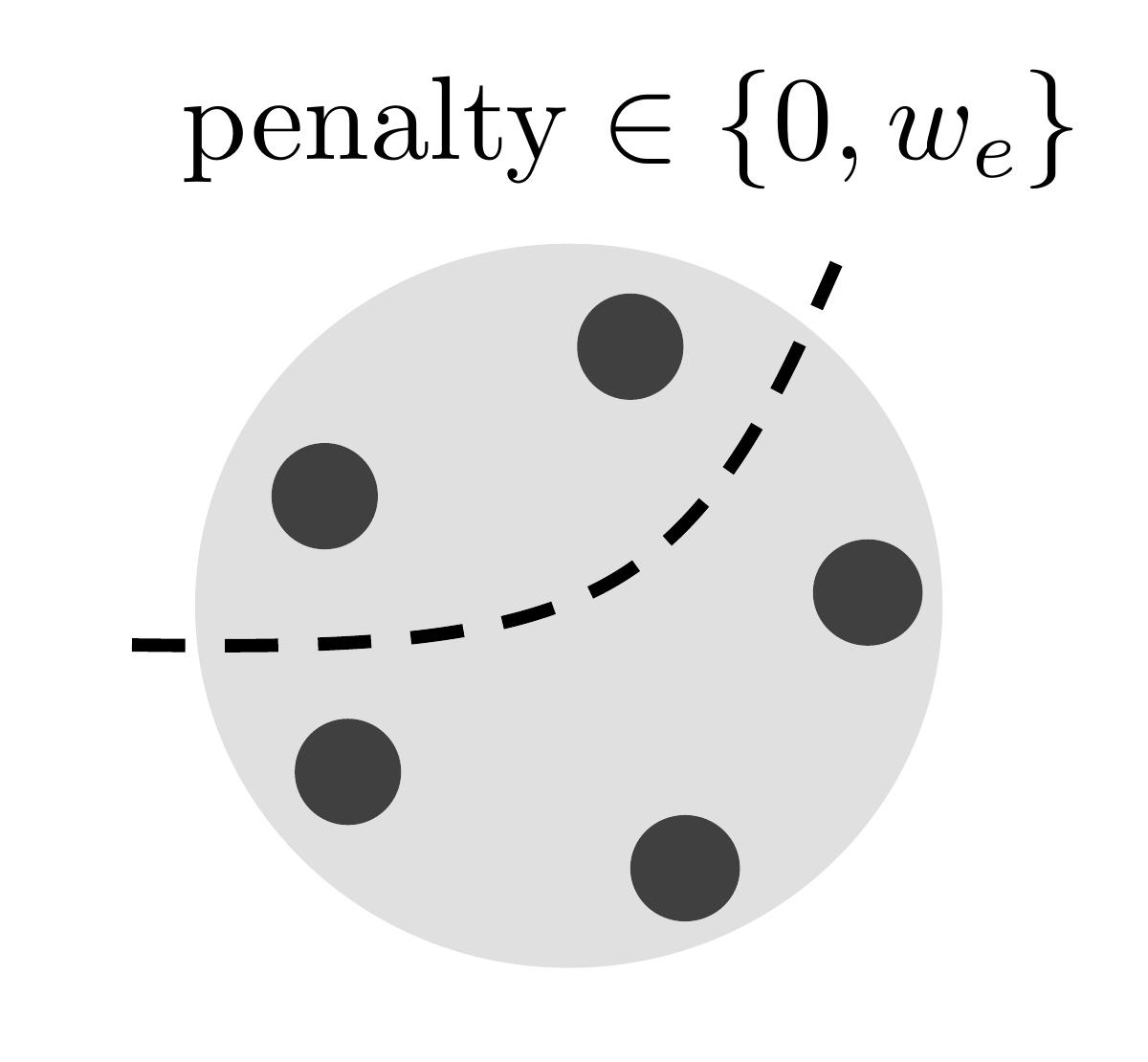}}\hfill
	\subfloat[Clique expansion \label{fig:clique}] 
	{\includegraphics[width=.33\linewidth]{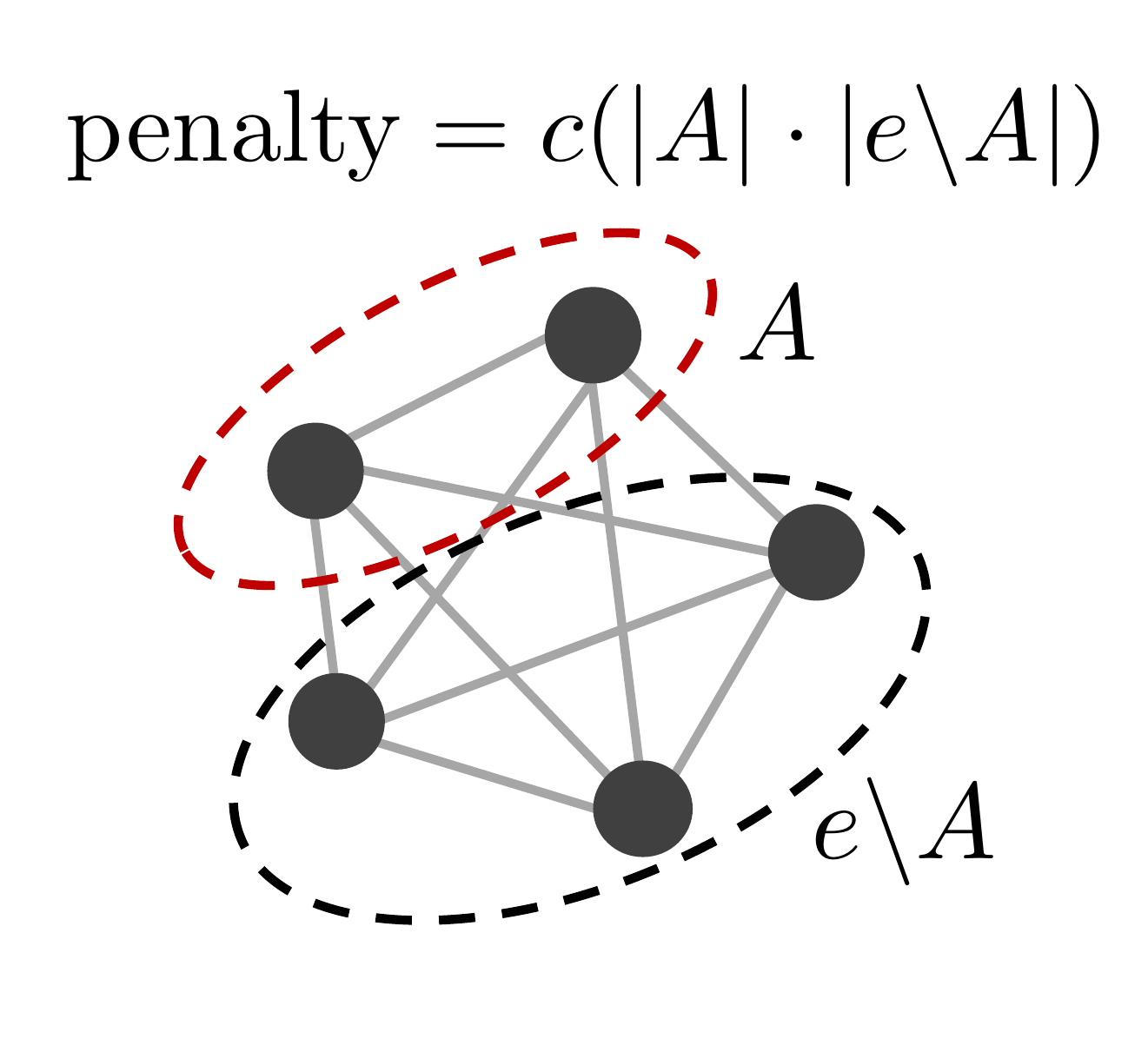}}
		\subfloat[General cut function \label{fig:gencut}] 
	{\includegraphics[width=.33\linewidth]{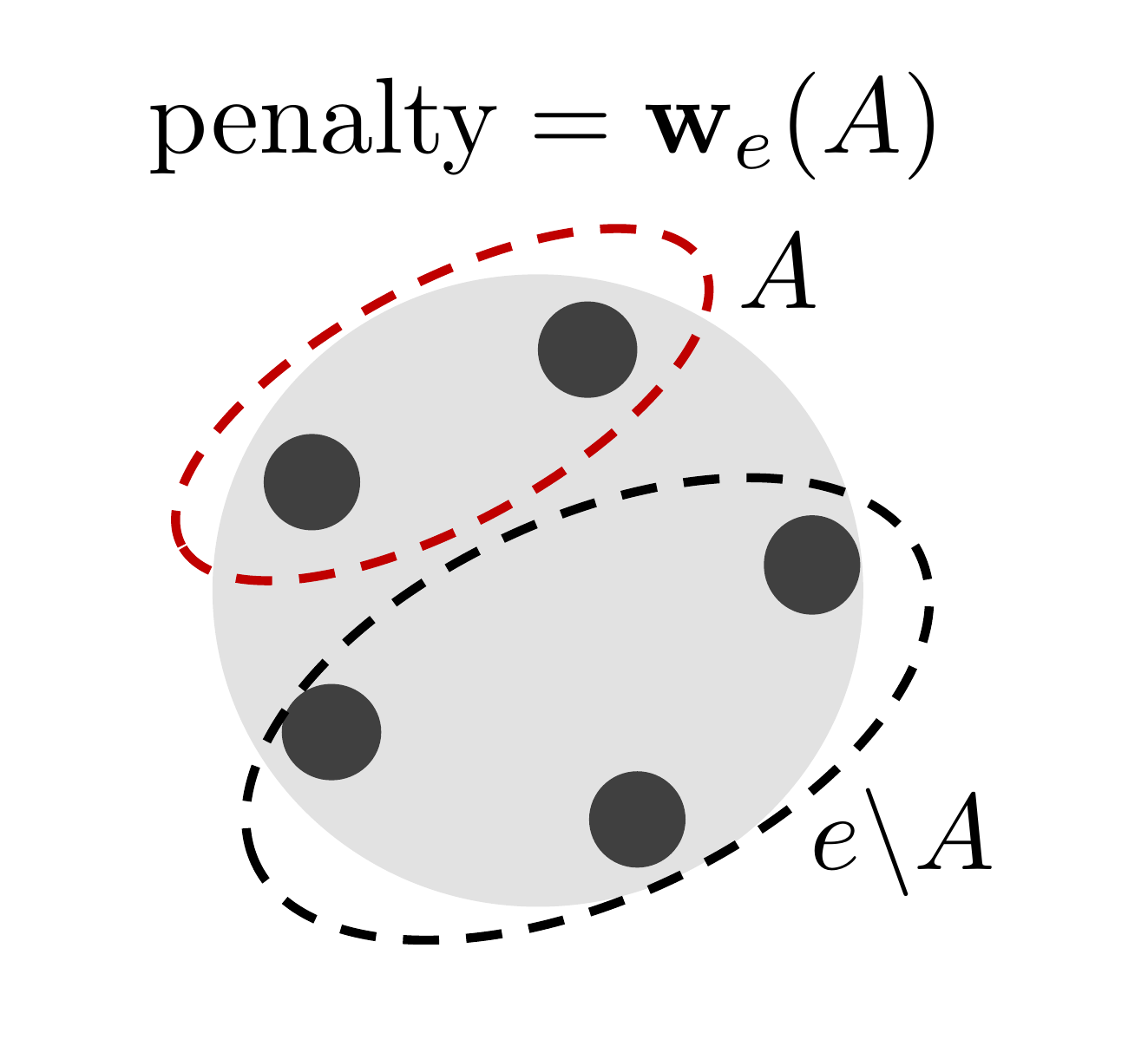}}
	\caption{Three different models for hyperedge penalties. The all-or-nothing cut has a scalar penalty $w_e$ for any way of cutting a hyperedge $e$; the penalty is zero if $e$ is uncut. The clique expansion converts a hyperedge to a clique, leading to a cut penalty proportional to the product of the set sizes that $e$ is split into. Recent generalized hyperedge cut functions assign penalties for each subset of nodes in $e$~\cite{panli2017inhomogeneous,panli_submodular,veldt2020hypercuts}.}
	\label{fig:cuts}
\end{figure}

Solving an all-or-nothing cut or applying clique expansion are only two specific models for higher-order relationships.
And if one uses a ratio cut model for clusters in a hypergraph, how to penalize a cut hyperedge may depend on the application.
Along these lines, \emph{inhomogeneous hypergraphs} model every possible way to split up a hyperedge to assign different penalties~~\cite{panli2017inhomogeneous};
however, these more sophisticated models are still approximated with weighted clique expansions.
In recent work, we developed a framework for \emph{exactly} solving hypergraph $s$-$t$ cut problems under similar notions of generalized hyperedge \emph{splitting functions}~\cite{veldt2020hypercuts} (\cref{fig:gencut}).
However, these techniques have not been applied to ratio cut objectives or to solve practical machine learning and data mining problems.

There is also little work on \emph{localized} hypergraph clustering, i.e., methods for finding well-connected sets of nodes that are biased towards a given region of a hypergraph.
Existing approaches~\cite{hao2017local,li2018tail} use random-walk-based local graph clustering methods~\cite{AndersenChungLang2006} on clique expansions.
\emph{Flow-based} methods are an alternative to random walks~\cite{Andersen:2008:AIG:1347082.1347154,Orecchia:2014:FAL:2634074.2634168,LangRao2004,veldt16simple,fountoulakis2020flowbased},
which have strong runtime and ratio cut quality guarantees.
These methods solve max.\ $s$-$t$ flow / min.\ $s$-$t$ cut problems as a subroutine.
Despite the success of these methods for graphs, they have not been extended to the hypergraph setting.

\xhdr{The present work: flow-based local hypergraph clustering}
Here, we develop a flow-based framework for local hypergraph clustering
based on minimizing localized ratio cut objectives. Our framework takes in a
set of input nodes and solves a sequence of hypergraph $s$-$t$ cut problems to
return a well-connected cluster of nodes that has a high overlap with the input
set, where ``well-connected'' is formalized with a ratio-cut-style objective
that can incorporate a wide range of hypergraph cut functions, including the
all-or-nothing penalty or the clique-expansion-penalty.

Unlike clique expansion techniques, we do not simply reduce to an existing graph technique.
Instead, our flow-based framework uses minimum \emph{hypergraph} $s$-$t$ cut computations.
Thus, we can leverage recent results on minimum $s$-$t$ hypergraph cut algorithms~\cite{veldt2020hypercuts}
to \emph{exactly} solve localized hypergraph ratio cut objectives for
generalized notions of hypergraph cuts. 
Importantly, although the $s$-$t$ hypergraph cut solver only relies on solving graph $s$-$t$ cut problems, 
the hypergraph ratio-cut is exactly optimized.
Our implementation can also make use of high-performance maximum $s$-$t$ flow solvers.

Our method comes with new guarantees on hypergraph conductance and
normalized cut. The conductance results generalize
previous guarantees in the graph setting to hypergraphs. The normalized
cut guarantees are the first of their kind; since graphs are a special case
of hypergraphs, we also have new results for the graph setting as a bonus. 
Our theory provides tighter guarantees than those
obtained by applying existing results for approximate recovery of
low-conductance sets, even in the graph setting.

A major feature of our methods is that they run in \emph{strongly-local time},
meaning that the runtime is dependent only on the size of the input set, rather than the entire hypergraph. 
Therefore, we can find optimal clusters without even seeing the entire hypergraph,
making our algorithms remarkably scalable in theory and practice.

We demonstrate our framework on large real-world hypergraphs, detecting product
categories in Amazon review data and similar questions on Stack Overflow data.
Given a small set of seed nodes, our method finds
clusters with thousands of nodes from hypergraphs with millions of nodes and
hyperedges and large average hyperedge size, often within a few seconds. 
Our methods are also more accurate than heuristics based on refining
neighborhoods of a seed set or using existing graph methods on
clique expansions.

%% file: v3-prelims.tex
\section{Preliminaries and Related Work}

\subsection{Background: Local Conductance in Graphs}
Let $G = (V,E)$ be an undirected graph and $w_{ij} \geq 0$ be the weight for edge $(i,j) \in E$.
The degree of a node $v$ is $d_v=\sum_{u \in N_v} w_{uv}$, where $N_v$ is the set of nodes sharing an edge with $v$.
A common graph clustering objective is conductance, defined for $S \subseteq V$ by
\begin{equation}
\label{conductance}
{\textstyle \cond_G(S)  = \frac{\cut(S)}{\min \{ \vol(S), \vol(\bar{S})\}},}
\end{equation}
where $\vol(S) = \sum_{v \in S} d_v$ is the volume of nodes in $S$, and $\cut(S) = \sum_{i \in S, j \in \bar{S}} w_{ij}$, which equals $\cut(\bar{S})$ by definition.
Here, $\bar{S} = V \backslash S$ is the complement set of $S$.
A related objective, which differs by at most a factor of two from conductance, is normalized cut~\cite{shimalik-ncut}:
\begin{equation}
\label{ncut}
{\textstyle \ncut_G(S)  = \frac{\cut(S)}{\vol(S)} + \frac{\cut(\bar{S})}{\vol(\bar{S})} = \vol(V)\frac{\cut(S)}{\vol(\bar{S}) \vol({S})}.}
\end{equation}

Conductance and normalized cut are both NP-hard to minimize~\cite{wagner1993between}.
However, localized variants can be minimized in polynomial time with repeated max.\ $s$-$t$ flow computations~\cite{Andersen:2008:AIG:1347082.1347154,LangRao2004,Orecchia:2014:FAL:2634074.2634168,veldt16simple,Veldt2019flow}.
For example, given nodes $R \subseteq V$ with $\vol(R) \leq \vol(\bar{R})$, the following objective can be minimized in polynomial time~\cite{LangRao2004}:
\begin{equation}
\label{mcs}
\minimize_{S\subseteq R} \,\, \cond_G(S) \,.
\end{equation}
In other words, given a region defined by a \emph{reference} set $R$, one can find the minimum conductance subset of $R$ in polynomial time, even though minimizing conductance over an entire graph is NP-hard. A more general \emph{local conductance} objective is
\begin{equation}
\label{localcond}
{\textstyle \lcond_{R,\varepsilon}(S) = 
\frac{\cut(S)}{\vol(S\cap R) - \varepsilon \vol(S \cap \bar{R})}\,.}
\end{equation}
Objective~\eqref{localcond} is minimized over all sets of nodes for which the denominator is positive (to avoid trivial outputs).
The denominator rewards sets $S$ that overlap with $R$.
The \emph{locality} parameter $\varepsilon$ controls the penalty for including nodes outside $R$.
As $\varepsilon \rightarrow \infty$, minimizing Eq.~\eqref{localcond} over sets producing a positive denominator
becomes equivalent to Eq.~\eqref{mcs}.
Andersen and Lang~\cite{Andersen:2008:AIG:1347082.1347154} showed how to minimize this objective for $\varepsilon = \vol(R)/\vol(\bar{R})$, and faster algorithms were later developed for $\varepsilon \gg \vol(R)/\vol(\bar{R})$~\cite{Orecchia:2014:FAL:2634074.2634168,veldt16simple}.
These algorithms repeatedly solve maximum flow problems on an auxiliary graph.

Minimizing objective~\eqref{localcond} can also provide cluster quality guarantees in terms of standard conductance (Eq.~\eqref{conductance}).
For example, the optimal set for \eqref{localcond} with $\varepsilon = \vol(R)/\vol(\bar{R})$ has conductance within a small factor of the conductance of \emph{any} set with a certain amount of overlap with $R$~\cite{Andersen:2008:AIG:1347082.1347154},
and there are related results for other values of $\varepsilon$ and objectives~\cite{Veldt2019flow,veldt16simple,Orecchia:2014:FAL:2634074.2634168}.
These are often called \emph{cut improvement} guarantees, since the ratio cut score of the output set \emph{improves} upon the score of the input $R$. See the work of Fountoulakis et al.~\cite{fountoulakis2020flowbased} for a detailed survey on flow-based methods for solving objective~\eqref{localcond}.

\subsection{Background: Generalized Hypergraph Cuts}
\label{sec:hyperst}
We now consider a hypergraph $\mathcal{H} = (V, E)$, where each edge $e \in E$ is a subset of $V$ (an undirected graph is then the special case where $\lvert e \rvert = 2$ for all $e \in E$).
A hypergedge $e \in E$ is \emph{cut} by a set $S \subseteq V$ if $e \cap S \neq \emptyset$ and $e \cap \bar{S} \neq \emptyset$, i.e., the hyperedge spans more than one cluster. We denote the set of edges cut by $S$ by $\partial S$. 
The most common generalization of graph cut penalties to hypergraphs is to assign no penalty if $e \in E$ is not cut,
but assign a fixed-weight scalar penalty of $w_e$ for any way of cutting $e$~\cite{lawler1973,ihler1993modeling,hadley1995}.
Inhomogeneous hypergraphs~\cite{panli2017inhomogeneous} and submodular hypergraphs~\cite{panli_submodular} generalize this by associating a weight \emph{function} with each edge, rather than a scalar; in this model, every distinct way of separating the nodes of a hyperedge can have its own penalty. Recently we considered these types of hyperedge weight functions, which we call \emph{splitting functions}, in the context of hypergraph $s$-$t$ cut problems~\cite{veldt2020hypercuts}.

We cover the splitting function terminology for general cut penalties here.
Each edge $e \in E$ is associated with a
\emph{splitting function} $\vw_e\colon A \subseteq e \rightarrow \mathbb{R}_{\ge 0}$ that maps each subset $A \subseteq e$ to a nonnegative splitting penalty.
Weights on edges can be directly incorporated into the splitting function $\vw_e$.
By definition, splitting functions are required to be symmetric and penalize only cut hyperedges:
\begin{align}
\label{eq:splitting_sym} 
\vw_e (A) = \vw_e (e \backslash A) \text{ and } \vw_e (e) = \vw_e(\emptyset) = 0 \,.
\end{align}
A number of different generalized hypergraph cut problems are known to be easier to solve or approximate when splitting functions are \emph{submodular}~\cite{panli2017inhomogeneous,panli_submodular,veldt2020hypercuts}. This means that for all $A\subseteq e$ and $B \subseteq e$,
\begin{equation}
\label{submodular}
\vw_e(A) + \vw_e(B) \geq \vw_e(A\cup B) + \vw_e(A\cap B)\,.
\end{equation}

A splitting function is \emph{cardinality-based} if it depends only on the number of nodes on each side of a split:
\begin{equation}
\label{cardbased}
\vw_e(A) = \vw_e(B) \hspace{0.3cm} \text{ whenever $|A| = |B|$}\,.
\end{equation}
Given a splitting function for each hyperedge, the generalized hypergraph cut penalty for a set $S \subseteq V$ is given by
\begin{align}
\textstyle \cut_\mathcal{H}(S) = \sum_{e \in E} \vw_e(e \cap S).
\end{align}
(By the symmetry constraint in Eq.~\eqref{eq:splitting_sym}, $\cut_\mathcal{H}(S) = \cut_\mathcal{H}(\bar{S})$.)
The generalized hypergraph $s$-$t$ cut objective is then:
\begin{align}
\label{eq:genhypstcut}
\minimize_{S \subseteq V} \,\,& \cut_\mathcal{H}(S), \text{ subject to } s \in S, t \in \bar{S}\,,
\end{align}
where $s$ and $t$ are designated source and sink nodes.
With the all-or-nothing splitting function, Eq.~\eqref{eq:genhypstcut} is solvable in polynomial time via reduction to a directed graph $s$-$t$ cut problem~\cite{lawler1973}.
More generally, if all splitting functions are submodular, the hypergraph $s$-$t$ cut problem is equivalent to minimizing a sum of submodular functions, which can be solved using general submodular function minimization~\cite{grotschel1981,Orlin2009,Schrijver:2000:CAM:361537.361552}
or specialty solvers for sums of submodular functions~\cite{kolmogorov2012minimizing,li2018revisiting,ene2017decomposable,stobbe2010efficient}.
Recently, we showed that when every splitting function is cardinality-based (Eq.~\eqref{cardbased}),
the hypergraph $s$-$t$ cut can be solved via reduction to a graph $s$-$t$ cut problem if and only if all splitting functions are submodular~\cite{veldt2020hypercuts}.
Cardinality-based submodular splitting functions are the focus of our models.

\subsection{Hypergraph Ratio Cut Objectives}
Given this general framework for hypergraph cuts, we present definitions for hypergraph conductance and normalized cut.
Let $\mathcal{H} = (V,E)$ be a hypergraph.
The hypergraph volume of $S \subseteq V$ is
\begin{equation}
\label{hypervolume}
\textstyle \vol_\mathcal{H}(S) = \sum_{v \in S} d_v\,,
\end{equation}
where $d_v = \sum_{e: v \in e} \vw_e(\{v\})$ is the hypergraph degree of $v$~\cite{panli2017inhomogeneous}. 
We define hypergraph conductance and normalized cut of $S \subseteq V$ as
\begin{align}
\textstyle \cond_\mathcal{H}(S) &= \textstyle \frac{\cut_\mathcal{H} (S)}{\min \{ \vol_\mathcal{H}(S), \vol_\mathcal{H}(\bar{S})\}}\label{hypercond} \\
\textstyle \ncut_\mathcal{H}(S) &= \textstyle \frac{\cut_\mathcal{H}(S)}{\vol_\mathcal{H}(S)} + \frac{\cut_\mathcal{H}(\bar{S})}{ \vol_\mathcal{H}(\bar{S})}\,.
\label{hyperncut}
\end{align}
When $\mathcal{H}$ is a graph, these reduce to the definitions of conductance~\eqref{conductance} and normalized cut~\eqref{ncut} in graphs.
The hypergraph conductance in Eq.~\eqref{hypercond} has been used with the all-or-nothing splitting function~\cite{BensonGleichLeskovec2016,chan2018generalizing,chan2018spectral}.
Eq.~\eqref{hyperncut} generalizes the version of hypergraph normalized cut from Zhou et al.~\cite{Zhou2006learning}, which is
the special case of
\begin{equation*}
\textstyle \vw_e(A) = \frac{w_e}{|e|} \cdot |A|\cdot |e \backslash A| \hspace{.2cm} \text{for all $A \subseteq e$,}
\end{equation*}
where $w_e$ is a scalar associated with $e \in E$.
These objectives correspond to the notions of hypergraph normalized cut and conductance considered by Li and Milenkovic~\cite{panli2017inhomogeneous,panli_submodular}.

%% file: v2-hypergraph-objectives.tex

\section{Hypergraph Local Conductance}
\label{sec:hlc}
We now define our new localized hypergraph ratio cut objectives. Let $\mathcal{H} = (V,E)$ be a hypergraph and $R$ a set of input nodes.
We define a function $\Omega_{R,\varepsilon}$, which measures the overlap between $R$ and another set of nodes $S$, parameterized by some $\varepsilon \geq \volh(R)/\volh(\bar{R})$. 
\begin{equation}
\label{overlap}
\Omega_{R,\varepsilon}(S) = \vol_\mathcal{H}(S\cap R) - \varepsilon \vol_\mathcal{H}(S \cap \bar{R})\,.
\end{equation}
This is our hypergraph analog to the denominator in Eq.~\eqref{localcond}.
To find a good cluster of nodes in $\mathcal{H}$ ``near'' $R$, we minimize
\begin{equation}
\label{hlc}
\textbf{HLC}_{R,\varepsilon}(S) = 
\begin{cases}
\frac{\cut_\mathcal{H}(S)}{\Omega_{R,\varepsilon}(S)} & \text{if $\Omega_{R,\varepsilon}(S) > 0$}\\
\infty & \text{otherwise,}
\end{cases}
\end{equation}
which we call hypergraph $(R,\varepsilon)$-localized conductance.
When $R$ and $\varepsilon$ are clear from context, we refer to~\eqref{hlc} as HLC (hypergraph localized conductance), 
denoting its value by $\textbf{HLC}(S)$.
This objective reduces to the graph case when the hypergraph is a graph.

In this section, we show how to minimize \hlc{}, given access to a minimum hypergraph $s$-$t$ cut solver
and consider cases where such solvers can be easily implemented with standard graph $s$-$t$ cut solvers.
\Cref{sec:strongly_local} shows how to optimize the procedure outlined in this section
(more formally, to have strongly-local runtime guarantees), and
\cref{sec:improve} adapts these results to provide bounds on the hypergraph conductance and normalized cut objectives.

\subsection{Minimizing the HLC Objective}
\label{sec:min_hlc}
We now provide a procedure that minimizes \hlc{}, given polynomially many queries to a solver for $\cut_\mathcal{H}$.
\Cref{sec:new_splitting} then considers cases where the solver itself requires polynomial time.
Let $\mathcal{H} = (V,E)$ be the original input hypergraph and $R$ the input set.
We minimize \hlc{} by repeatedly solving hypergraph min.\ $s$-$t$ cut problems on an \emph{extended} hypergraph $\mathcal{H}_\alpha$, parameterized by $\alpha \in (0,1)$:
\begin{itemize}
	\item Keep all of $\mathcal{H} = (V,E)$ with original splitting functions.
	\item Introduce a source node $s$ and sink node $t$.
	\item For each $r \in R$, add an edge $(s,r)$ with weight $\alpha d_r$. 
	\item For each $j \in \bar{R}$, add an edge $(j,t)$ with weight $\alpha \varepsilon d_j$.
\end{itemize} 
By construction, $\mathcal{H}_\alpha$ contains hyperedges from $\mathcal{H}$ and 
pairwise edges attached to source and sink nodes.
Each hyperedge $e$ in $\mathcal{H}$ is associated with a splitting function $\vw_e$.
We call edges adjacent to $s$ and $t$ \emph{terminal} edges, and these have standard cut penalty: 
0 if the edge is not cut, otherwise the penalty is the weight of the edge. 
For any $S \subseteq V$, the value of the hypergraph cut $S \cup \{s\}$ in $\mathcal{H}_\alpha$ is
\begin{equation}
\label{hypergraphcut}
\textbf{H-st-cut}_\alpha(S) = \cut_\mathcal{H}(S) + \alpha \vol_\mathcal{H}(\bar{S} \cap R) + \alpha \varepsilon \vol_\mathcal{H}(S \cap \bar{R}).
\end{equation}

Choosing $S = \emptyset$ gives an upper bound of $\alpha \vol_\mathcal{H}(R)$ on the minimum cut score.
Thus, if the minimizer $S^*$ for Eq.~\eqref{hypergraphcut} has a cut score strictly less than $\alpha \vol_\mathcal{H}(R)$,
then $S^*$ must be nonempty, and we can rearrange~\eqref{hypergraphcut} to show that
\begin{equation}
	\hlc(S^*) = \frac{\cut_\mathcal{H}(S^*)}{\vol_\mathcal{H}(S^*\cap R) - \varepsilon \vol_\mathcal{H}(S^* \cap \bar{R})} < \alpha \,.
\end{equation}
Thus, for any $\alpha \in (0,1)$, to find out if there is a nonempty $S \subseteq V$ with HLC value less than $\alpha$,
it suffices to solve a generalized hypergraph $s$-$t$ cut problem.
\Cref{alg:1} gives a procedure for minimizing \textbf{HLC}, based on repeatedly solving objective~\eqref{hypergraphcut} for smaller values of $\alpha$ until no more improvement in the HLC objective is possible. 
\begin{algorithm}[t]
	\caption{Hypergraph local conductance minimization.}
	\label{alg:1}
	\begin{algorithmic}
		\State \textbf{Input:} $\mathcal{H}$, $R$, $\varepsilon \geq \volh(R)/\volh(\bar{R})$, $\cut_\mathcal{H}$.
		\State Set $\alpha = \hlc(R)$ and $S = R$
		\Do
		\State Update $S_{\text{best}} \leftarrow S$ and save $\alpha_0 \leftarrow \alpha$
		\State $S \leftarrow \arg \min_{S'} \textbf{H-st-Cut}_{\alpha}(S')$ 
		\State $\alpha \leftarrow \hlc(S)$
		\doWhile{$\alpha < \alpha_0$}
		\State \textbf{Return:} $S_{\text{best}}$
	\end{algorithmic}
\end{algorithm}



\subsection{A New Hyperedge Splitting Function}\label{sec:new_splitting}
\Cref{alg:1} repeatedly solves the $\textbf{H-st-cut}$ objective~\eqref{hypergraphcut} in an \emph{auxiliary} hypergraph $\mathcal{H}_\alpha$.
For submodular splitting functions, this can be done in polynomial time with methods for minimizing sums of submodular functions~\cite{kolmogorov2012minimizing,li2018revisiting,ene2017decomposable,stobbe2010efficient}. For the more restrictive class of cardinality-based submodular splitting functions, we only need to solve directed graph $s$-$t$ cut problems~\cite{veldt2020hypercuts}.
Implementations of such solvers are readily available and perform well in practice.

As an example, we present a new class of cardinality-based splitting functions that depends on a single tunable integer parameter $\delta \geq 1$,
which we use for our numerical experiments for its modeling capability and computational appeal:
\begin{equation}
\label{eq:tl}
\vw_e(A) = \min \{ \delta , |A|, |e\backslash A| \} \hspace{0.4cm} \text{ for any $A \subseteq e$.}
\end{equation}
We call this the \emph{\deltasplit{}} splitting function, since the penalty is linear in terms of the small side of the cut, up until a maximum penalty of $\delta$.
In other words, in a split hyperedge, we incur a unit cost for adding another node to the small side of the cut, up until we reach $\delta$ such nodes.
The $\delta = 1$ case is equivalent to the unweighted all-or-nothing cut.
For large enough $\delta$, the \deltasplit{} is the linear hypergedge splitting penalty, 
which is equivalent to applying a star expansion to the hypergraph~\cite{zien1999}.
Choosing different values for $\delta$ interpolates between these common splitting functions, which
enables the detection of different types of cut sets in a hypergraph and provides the data modeler with flexibility.

We now show how to efficiently optimize $s$-$t$ hypergraph cuts for this splitting function using graph $s$-$t$ cuts.
Let $\mathcal{H} = (V,E)$ be a hypergraph, where each edge is associated with the \deltasplit{} splitting function.
A minimum $s$-$t$ cut problem in $\mathcal{H}$ can be reduced to a minimum $s$-$t$ cut problem in a new directed graph $G_\mathcal{H}$ by replacing each $e \in E$ with the following \emph{gadget}:
\begin{itemize}
\item Introduce two auxiliary nodes $v_e'$ and $v_e''$.
\item Create a directed edge from $v_e'$ to $v_e''$ with weight $\delta$.
\item For each $v \in e$, add directed edges $(v, v_e')$ and $(v_e'', v)$, both with weight 1.
\end{itemize}
In any min.\ $s$-$t$ cut solution in $G_\mathcal{H}$, the auxiliary nodes of $e$ 
are arranged in a way that leads to a minimum possible cut. 
If $e$ is cut and $A \subseteq e$ is on the source side of the cut, 
then the penalty in $G_\mathcal{H}$ removes all directed paths from $A$ to $e\backslash A$ 
by taking the smaller penalty among three options: 
cutting (i) the middle edge $\delta$, (ii) all edges from $A$ to $v_e'$, or (iii) all edges from $v_e''$ to $e\backslash A$.
Thus, the penalty in $G_\mathcal{H}$ at this gadget will be exactly the hypergraph splitting penalty~\eqref{eq:tl}.
This gadget-based approach is related to our
recent techniques for modeling general cardinality-based submodular splitting functions~\cite{veldt2020hypercuts}.
However, the approach here requires fewer auxiliary nodes and directed edges.
In Section~\ref{sec:experiments}, we show that minimizing HLC with the 
\deltasplit{} penalty for varying $\delta$ enables us to
efficiently detect better clusters in a large hypergraph.

Finally, we provide an upper bound on the number of minimum $s$-$t$ cut problems Algorithm~\ref{alg:1} must solve if the \deltasplit{} penalty is applied. A proof is given in the appendix.
\begin{theorem}
	\label{thm:itbound}
	Let $S_i$ be the set returned after the $i$th iteration of Alg.~\ref{alg:1}. 
	The value of $\cut_\mathcal{H}(S_i)$ strictly decreases until the last iteration. 
	Thus, with the \deltasplit{} penalty, the cut value decreases by at least one in each iteration, for a maximum of $\cut_\mathcal{H}(R)$ iterations.
\end{theorem}
This bound is loose in practice --- the algorithm converges in 2--5 iterations in nearly all of our experiments.
Similar results can be developed for non-integer $\delta$, though we omit the details.

%% file: v3-strong-locality.tex

\section{A Strongly-Local Algorithm}\label{sec:strongly_local}
Recall that $\varepsilon$ is a \emph{locality} parameter that controls the sets $S \subseteq V$ for which $\hlc(S) < \infty$. 
If $\varepsilon$ is large, then $\Omega_{R,\varepsilon}(S) < 0$ for many sets that do not substantially overlap with $R$. 
In local clustering applications where $\vol(R)$ is small compared to the entire hypergraph,
there is little gained by considering sets $S$ that share little overlap with $R$. 
In these settings, it is preferable to explore only a small region \emph{nearby} $R$, 
rather than exploring every node and hyperedge in the hypergraph.
Thus, it is natural to choose a larger value of $\varepsilon$ and output low hypergraph conductance sets overlapping with $R$.
Ideally, we want to avoid even looking at the entire hypergraph.

We can formalize this idea via strong locality.
A local clustering algorithm is \emph{strongly-local} if its runtime depends only on the size of the input set and not the entire hypergraph.
In constrast, Alg.~\ref{alg:1} is \emph{weakly-local}, meaning that its output is biased towards a region of the hypergraph, 
but its runtime may still depend on the size of the entire hypergraph.
In this section, we give a strongly-local variant of Alg.~\ref{alg:1}, when $\vol(R) \ll \vol(\bar{R})$ and 
$\varepsilon \gg \vol(R) /\vol(\bar{R})$ is treated as a small constant.
Our procedure generalizes strongly-local methods for minimizing local conductance in graphs~\cite{Orecchia:2014:FAL:2634074.2634168,veldt16simple}.
For the hypergraph setting, we combine previous techniques for local min.\ $s$-$t$ cut computations in graphs with
techniques for converting hypergraph $s$-$t$ cut problems into graph $s$-$t$ cut problems~\cite{lawler1973,veldt2020hypercuts}.

\subsection{Making the Procedure Strongly-Local}
In order to solve the hypergraph $s$-$t$ cut objective~\eqref{hypergraphcut} in strongly-local time, 
we must avoid explicitly constructing $\mathcal{H}_\alpha$. 
We instead begin with a sub-hypergraph $\mathcal{L}$ of $\mathcal{H}_\alpha$, which we call the \emph{local hypergraph}, and alternate between the following two steps:
\begin{enumerate}
	\item Solve a hypergraph minimum $s$-$t$ cut problem on $\mathcal{L}$.
	\item Grow the subgraph $\mathcal{L}$ based on the $s$-$t$ cut solution.
\end{enumerate}
The algorithm proceeds until a convergence criterion is satisfied, at which
point the growth mechanism in Step 2 will stop and the algorithm will
output the minimum $s$-$t$ cut solution for $\mathcal{H}_\alpha$.

\xhdr{Algorithm Terminology}
Let $\mathcal{H} = (V,E)$ be the hypergraph. 
For any $v \in V$, let $E(v) = \{e \in E : v \in e\}$, and define $E(S) = \cup_{v \in S} E(v)$ for any set $S \subseteq V$. Let $\mathcal{H}_\alpha = (V\cup \{s,t\}, E \cup E^{st})$ be defined as in \cref{sec:hlc},
where $E^{st}$ is the terminal edge set.
For $v \in V$, let $e_v^{st}$ denote its terminal edge in $\mathcal{H}_\alpha$ (recall that nodes in $R$ are connected to $s$ and nodes in $\bar{R}$ are connected to $t$),
and let $E^{st}(S)$ denote the set of edges between nodes in $S \cup \{s,t\}$ for any set $S$.
Our goal is to find a min.\ $s$-$t$ cut of $\mathcal{H}_\alpha$ without forming $\mathcal{H}_\alpha$. 
To do so, we assume that we have an oracle that efficiently outputs $E(v)$ for any $v \in V$. 
From $E(v)$, we can extract the \emph{neighborhood} of $v$ in $\mathcal{H}$:
 \begin{equation}
 \mathcal{N}(v) = \{ u \in V \;\vert\; \exists e \in E \text{ such that } u,v \in e\}
 \end{equation} 
For a set $S$, we define $\mathcal{N}(S) = \cup_{v \in S} \mathcal{N}(v)$. 
 
\xhdr{The Local Hypergraph}
Let $\mathcal{L} = (V_L \cup \{s,t\}, E_L\cup E_L^{st})$ denote the {local hypergraph}, a growing subgraph of $\mathcal{H}_\alpha$ over which we will repeatedly solve min.\ $s$-$t$ cut problems. 
We initialize $\mathcal{L}$ to contain all nodes and neighbors of $R$, i.e., $V_L = R \cup \mathcal{N}(R)$, and add the terminal edge for each of these nodes to $E_L^{st}$. The set $E_L$ is initialized to $E(R)$, the set of edges containing at least one node from $R$. As the algorithm progresses, $\mathcal{L}$ grows to include more nodes and edges from $\mathcal{H}_\alpha$, always maintaining that $V_L \subseteq V$, $E_L \subseteq E$, and $E_L^{st} \subseteq E^{st}$. 
For $S \subseteq V_L$, let $\textbf{L-st-cut}_\alpha(S)$ be the value of the $s$-$t$ cut $S\cup \{s\}$ in $\mathcal{L}$. 
Because $\mathcal{L}$ is a sub-hypergraph of $\mathcal{H}_\alpha$, at every step we have that
\begin{equation}
\label{lleh}
\textbf{L-st-cut}_\alpha(S) \leq \textbf{H-st-cut}_\alpha(S) \text{ for all $S \subseteq V_L$}.
\end{equation}
By carefully selecting new nodes and edges to add to $\mathcal{L}$, we will show that the min.\ $s$-$t$ cut solution in $\mathcal{L}$ converges to the min.\ $s$-$t$ cut solution in $\mathcal{H}_\alpha$, without ever fully forming $\mathcal{H}_\alpha$.

 
\xhdr{Two Alternating Steps}
After initializing $\mathcal{L}$, we repeat two steps until convergence: (1) find a min.\ $s$-$t$ cut in $\mathcal{L}$, and 
(2) grow $\mathcal{L}$ based on the solution to the cut problem. 
To grow $\mathcal{L}$ at each iteration, we track which nodes from $\bar{R}$ have had their edge to the sink cut by some 
min.\ $s$-$t$ cut in $\mathcal{L}$ in a previous iteration. 
When $(v, t)$ is cut for the first time,
we expand the local hypergraph by adding the neighbors and edges adjacent to $v$ in $\mathcal{H}_\alpha$:
\begin{itemize}
	\item Update $V_L \leftarrow V_L \cup\mathcal{N}(v)$.
	\item Update $E_L \leftarrow E_L \cup E(v)$ and $E_L^{st} \leftarrow E_L^{st} \cup E^{st} (\mathcal{N}(v))$.
\end{itemize}
At this point, we say that node $v$ has been \emph{explored}, and we maintain a set of nodes $X$ that have been explored at any point during the algorithm. Since $R$ is already contained in $\mathcal{L}$, any new nodes we add to the local hypergraph will be from $\bar{R}$. 
\Cref{alg:local} shows pseudo-code for the overall procedure.

\begin{algorithm}[t]
	\caption{Strongly-Local Min $s$-$t$ cut solver}
	\label{alg:local}
	\begin{algorithmic}
		\State \textbf{Input}: $R$, $\varepsilon$, $\alpha$, and access to $E(v)$ for any $v$ in hypergraph $\mathcal{H}$.
		\State \textbf{Output}: Min $s$-$t$ cut solution $S$ for $\mathcal{H}_\alpha$, objective~\eqref{hypergraphcut}
		\State \texttt{// Initialize Local graph $\mathcal{L}$}
		\State $V_L \leftarrow R\cup \mathcal{N}(R)$, $E_L^{st} \leftarrow   E^{st}(V_L)$, $E_L \leftarrow E(R)$, $X \leftarrow \emptyset$, $N \leftarrow \emptyset$
		\Do
		\State \texttt{// \textbf{Step 1}: Solve a minimum $s$-$t$ cut problem on $\mathcal{L}$.}
		\State $S_L = \argmin_{S \subseteq V_L} \textbf{L-st-cut}_\alpha(S)$
		\State $N = S_L \cap \bar{R} \cap V\backslash X$  (nodes around which to expand $\mathcal{L}$)
		\State \texttt{// \textbf{Step 2}: Grow $\mathcal{L}$.}
		\State $V_L \leftarrow V_L \cup \mathcal{N}(N)$,\; $E_L \leftarrow E_L \cup E(N)$,\; $E_L^{st} \leftarrow E_L^{st} \cup E^{st}(\mathcal{N}(N))$
		\State $X \leftarrow X \cup N$.
		\doWhile{$N \neq \emptyset$}
		\State Return $S_L$
	\end{algorithmic}
\end{algorithm}

\xhdr{Convergence and Locality}
The algorithm terminates when, after a min.\ $s$-$t$ cut computation, no new edges are added $E_L$. 
At this point, the min.\ $s$-$t$ cut set in $\mathcal{L}$ is the min.\ $s$-$t$ cut set of $G_\alpha$. 
\begin{theorem}
	\label{thm:converges}
	The set $S$ returned by Algorithm~\ref{alg:local} minimizes objective~\eqref{hypergraphcut}, the minimum $s$-$t$ cut objective on $\mathcal{H}_\alpha$.
\end{theorem}

Furthermore, the following theorem shows that under reasonable conditions, the growth of the local hypergraph is bounded in terms of the $\vol_\mathcal{H}(R)$. Thus, our algorithm is strongly-local. We use the term \emph{graph-reducible} to refer to any hypergraph cut function for which the hypergraph $s$-$t$ cut problem can be reduced to an equivalent $s$-$t$ cut problem in a directed graph.
\begin{theorem}
	\label{thm:bounds}
	Suppose we have a cardinality-based submodular splitting function scaled with minimum non-zero penalty $1$ (e.g., \deltasplit{})
        and that no nodes in $R$ are isolated.
        Then the local hypergraph $\mathcal{L}$ will have at most $\frac{3}{2}(1 + 1 / \varepsilon) \vol_\mathcal{H}(R)$ hyperedges,
        and the number of nodes and terminal edges will both be at most 
	$k \vol_\mathcal{H}(R)(1 + 1/\varepsilon)$, where $k$ is the maximum size hyperedge in $\mathcal{H}_\alpha$.
\end{theorem}
Proofs of Theorems~\ref{thm:converges} and~\ref{thm:bounds} are in the
appendix. The minimum value on the splitting function is just a scaling issue,
and the cardinality-based submodular restriction lets us bound set sizes by volumes of those sets.
The assumption that $R$ has no isolated nodes is minor;
these nodes could be removed in a pre-processing step.

\subsection{Runtime Analysis}
\Cref{thm:bounds} gives strongly local runtimes for Alg.~\ref{alg:local} when we use the \deltasplit{} splitting function with $\delta \geq 1$, since it is a cardinality-based, submodular function where the minimum non-zero penalty is $1$. This immediately implies the same runtime result for the all-or-nothing penalty ($\delta = 1$).
We implement \cref{alg:local} for the \deltasplit{} penalty by replacing each edge $e$ with a small directed graph as outlined in \cref{sec:new_splitting}.
If $k$ is the maximum hyperedge size, this graph reduction introduces $2$ auxiliary nodes and at most $(2k + 1)$ 
directed edges for each hyperedge that appear in the local hypergraph.
Combining this with \cref{thm:bounds}, the graph reduction of the largest local hypergraph $\mathcal{L}$ has at most
\[
k\vol_\mathcal{H}(R)(1+ 1/\varepsilon) + 3 (1+1/\varepsilon) \vol_\mathcal{H}(R) = O(k\vol_\mathcal{H}(R)(1+ 1/\varepsilon))
\]
nodes and
\[
k\vol_\mathcal{H}(R)(1+ 1/\varepsilon) + (2k + 1)\frac{3}{2} (1+ 1/\varepsilon) \vol_\mathcal{H}(R) = O(k\vol_\mathcal{H}(R)(1+ 1/\varepsilon))
\]
edges.
For a graph $G = (V,E)$, there is an $s$-$t$ cut algorithm with runtime $O(|V||E|)$~\cite{Orlin:2013:MFO:2488608.2488705}.
In theory, we can use this to solve the hypergraph $s$-$t$ cut problem with \deltasplit{} penalties on the largest local hypergraph $\mathcal{L}$ in time $O(k^2 \vol_\mathcal{H}(R)^2(1+ 1/\varepsilon)^2)$.
The local hypergraph grows by at least one hyperedge (i.e., $(2k+1)$ directed edges) each step, so
we need to solve $O(k \vol_\mathcal{H}(R)(1 + 1/\varepsilon))$ $s$-$t$ cut problems, for an overall runtime of $O(k^3 \vol_\mathcal{H}(R)^3(1+ 1/\varepsilon)^3)$. 


Using high-performance max-flow/min-cut solvers, the runtime of our algorithm is fast in practice --- on
hypergraphs with millions of nodes and edges, roughly a few seconds for small $R$ and a few minutes for large $R$.
Nevertheless, as long as $1/\varepsilon$ is independent of the size of the input hypergraph 
(e.g., $\varepsilon = 1$ is always a valid choice),
then our overall procedure for minimizing localized ratio cuts in hypergraphs is strongly-local.
Furthermore, if the maximum hyperedge size $k$ is a constant, our asymptotic runtime is the same as the runtime for strongly-local \emph{graph} clustering algorithms~\cite{veldt16simple,Veldt2019flow}, which are also effective in practice.
Finally, although our analysis focused on the \deltasplit{} penalty used for our experiments, we can get
a similar runtime for any cardinality-based submodular splitting function, with a slightly worse dependence on $k$.

%% file: v1-cut-improvement.tex

\section{Ratio Cut Improvement Guarantees}\label{sec:improve}
The HLC objective is meaningful in its own right, and we use it in our experiments;
however, understanding the relationship between HLC and more standard ratio cut objectives that do not inherently depend on $R$ and $\varepsilon$ is also useful.
To this end, we derive guarantees satisfied by \cref{alg:1} and the HLC objective in terms of hypergraph conductance and normalized cut.
Our theory shows that the algorithm output has a ratio cut score that is nearly as good as any other set of nodes that have some overlap with $R$.

Our results in this section are for a fixed input hypergraph $\mathcal{H}$, so we drop the subscript $\mathcal{H}$ to simplify notation.
Throughout this section, let $\varepsilon_0 = \vol(R)/\vol(\bar{R}) \leq 1$ denote the minimum value of the locality parameter.
Setting $\varepsilon = \varepsilon_0$ gives the best cut improvement guarantees, which is always a valid choice.
However, we prove results for more general parameter settings, since, as discussed in \cref{sec:strongly_local}, there are locality and runtime benefits for setting $\varepsilon > \varepsilon_0$. 

\subsection{Hypergraph Conductance Guarantees}
We first generalize previous conductance improvement guarantees for local graph clustering~\cite{Andersen:2008:AIG:1347082.1347154,veldt16simple} to our hypergraph objective.
\begin{theorem}
	\label{thm:cond}
	Let $S^*$ be the set returned by \cref{alg:1} for some $\varepsilon \in (\varepsilon_0, \varepsilon_0 + 1)$, and let $\mu = \varepsilon  - \varepsilon_0 \geq 0 $.
	\begin{enumerate}
		\item For any $T \subseteq R$,  $\cond(S^*) \leq \cond(T).$
		\item For any set $T$ satisfying $\vol(T) \leq \vol(\bar{T})$ and
		\begin{equation}
		\label{assumption1}
		{\textstyle \frac{\vol(T \cap R)}{\vol(T)} \geq \frac{\vol(R)}{\vol(V)} + \gamma \frac{\vol(\bar{R})}{\vol(V)},}
		\end{equation}
		for some $\gamma \in (\mu,1)$, we have that $\cond(S^*) \leq \frac{1}{\gamma -\mu} \cond(T)$.
		Specifically, when $\varepsilon = \varepsilon_0$, $\cond(S^*) \leq (1/\gamma) \cond(T)$.
	\end{enumerate}
\end{theorem}
\begin{proof}
  For notational compactness, let $\textbf{v}(S) = \vol(S)$ for a set $S \subset V$. We first prove that $\Omega_{R,\varepsilon}(S) \leq \min\{\vlm(S), \vlm(\bar{S})\}$ for any set $S$,
  which implies that $\cond(S) \leq \hlc(S)$:
  \begin{align*}
      \Omega_{R,\varepsilon}(S) &= \vlm(S \cap R) - \varepsilon \vlm(S \cap \bar{R}) \leq \vlm(S \cap R) \leq  \vlm(S), \\[1mm]
	\Omega_{R,\varepsilon}(S) &\leq 	\vlm(S \cap R) - \varepsilon_0\vlm(S \cap \bar{R})\\
	&=\vlm(R) - \vlm(\bar{S} \cap R) - \varepsilon_0 (\vlm(\bar{R}) - \vlm(\bar{S} \cap \bar{R}))\\
	&\leq \vlm(R) - \varepsilon_0 \vlm(\bar{R}) + \varepsilon_0\vlm(\bar{R} \cap \bar{S}) = \varepsilon_0\vlm(\bar{R} \cap \bar{S}) \leq \vlm(\bar{S}),
	\end{align*}
        where the last equality uses the definition of $\varepsilon$ and the final inequality uses $\varepsilon_0 < 1$.
	Finally, for any $T \subseteq R$ where $\vlm(R) \le \vlm(\bar{R})$, $\hlc(T) = \cond(T)$, which gives the first theorem statement:
        $\cond(S^*) \leq \hlc(S^*) \leq \hlc(T) = \cond(T)$.
	
	For the second statement, if $\Omega_{R,\varepsilon}(T) \leq 0$, then $\hlc(T) = \infty$ and the result is trivial. Assume then that $\Omega_{R,\varepsilon}(T) > 0$. 
	Because $\hlc(S^*) \leq \hlc(T)$, the result will hold if we can prove that $\hlc(T) \leq \frac{1}{\gamma - \mu}\cond(T)$,
        which is true as long as $\Omega_{R,\varepsilon}(T) \geq (\gamma - \mu)\vlm(T)$. We prove this by applying assumption~\eqref{assumption1}. 

	\begin{align*}
	\frac{\Omega_{R,\varepsilon}(T)}{\vlm(T)} 
	&\geq \frac{(1+\varepsilon)\vlm(T \cap R)  - \varepsilon \vlm(T)}{\vlm(T)}
	\geq (1 + \varepsilon)\left(\frac{\vlm(R)}{\vlm(V)} + \gamma \frac{\vlm(\bar{R})}{\vlm(V)}  \right) - \varepsilon \\
	&= \left(1+ \frac{\vlm(R)}{\vlm(\bar{R})}\right)\gamma \frac{\vlm(\bar{R})}{\vlm(V)} + \mu\gamma \frac{\vlm(\bar{R})}{\vlm(V)} + (1+\varepsilon) \frac{\vlm(R)}{\vlm(V)}  - \varepsilon\\
	&= \gamma + \mu\gamma \frac{\vlm(\bar{R})}{\vlm(V)} + (1+\varepsilon) \frac{\vlm(R)}{\vlm(V)}  - \varepsilon\left( \frac{\vlm(R)}{\vlm(V)} + \frac{\vlm(\bar{R})}{\vlm(V)}\right)\\
	&= \gamma + (\mu \gamma - \mu)\frac{\vlm(\bar{R})}{\vlm(V)} - \frac{\vlm(R)}{\vlm(\bar{R})} \frac{\vlm(\bar{R})}{\vlm(V)} + \frac{\vlm(R)}{\vlm(V)} \geq \gamma  - \mu \,.
	\end{align*}
\end{proof}

\subsection{Hypergraph Normalized Cut Guarantees}
Given that conductance and normalized cut differ by at most a factor of two, 
we can translate \cref{thm:cond} into bounds for normalized cut.
Our next theorem, however, obtains better guarantees by directly developing bounds for hypergraph normalized cut.
This demonstrates how our algorithmic framework transcends its relationship with conductance,
as it can be used to find sets that also satisfy strong guarantees for other common objectives.
\begin{theorem}
	\label{thm:ncut}
	Let $S^*$ be the set returned by Algorithm~\ref{alg:1} for some $\varepsilon \in (\varepsilon_0, \varepsilon_0 + 1)$, and let $\mu = \varepsilon  - \varepsilon_0 \geq 0$.
	 If a set $T \subset V$ satisfies $\vol(T) \leq \vol(\bar{T})$ and for some $\beta \in (2\mu/(1+2\mu),1)$ satisfies
		\begin{equation}
                  \label{assumption2}
                  {\textstyle \frac{\vol(T \cap R)}{\vol(T)} \geq \frac{\vol(\bar{T} \cap R)}{\vol(\bar{T})} + \beta,}
                \end{equation}
          we have that $\ncut(S^*) \leq \frac{1}{\beta + 2\mu \beta -2\mu} \ncut(T)$.
          Specifically, when $\varepsilon = \varepsilon_0$, $\ncut(S^*) \leq (1/\beta)\ncut(T)$.
\end{theorem}
A full proof is in the appendix. The overlap assumptions~\eqref{assumption1} and~\eqref{assumption2} in \cref{thm:cond,thm:ncut} differ.
The assumption in \cref{thm:cond} matches previous local graph clustering results~\cite{Andersen:2008:AIG:1347082.1347154,veldt16simple}
and measures how much $R$ overlaps with a set $T$.
In contrast, assumption~\eqref{assumption2} provides a more intuitive measure of how much more $R$ overlaps with $T$ than it does with $\bar{T}$.
This is the first application of this type of overlap assumption for cut improvement --- graph or hypergraph.
We next give a simple example for how this overlap assumption and \cref{thm:ncut}
provide meaningful new normalized cut improvement guarantees, even in the well-studied graph setting.

\xhdr{Example}
Consider a hypergraph (or graph) $\mathcal{H} = (V,E)$ containing a low-conductance target set $T$ with $\vol(T) = \vol(\bar{T}) = \vol(V)/2$. Assume that we do not know all of $T$, but we have access to a set $R$ constituting half the volume of $T$, i.e., $R \subset T$ with $\vol(R) = \vol(T)/2$.
Let $S^*$ be the set returned by \cref{alg:1} when $\varepsilon = \vol(R)/\vol(\bar{R})$.
First, we apply \cref{thm:cond} to bound $\cond(S^*)$, where assumption~\eqref{assumption1} holds with $\gamma = 1/3$:
\begin{equation}
{\textstyle \frac{\vol(R\cap T)}{\vol(T)} = \frac{1}{2} = \frac{1}{4} + \frac{1}{3}\cdot\frac{3}{4} = \frac{\vol(R)}{\vol(V)} + \gamma \frac{\vol(\bar{R})}{\vol(V)},}
\end{equation}
and this value of $\gamma$ gives the tightest bound.
Thus, \cref{thm:cond} guarantees that $\cond(S^*) \leq 3\, \cond(T)$.
Next, using the fact that $\cond(S^*) \leq \ncut(S^*) \leq 2\cond(S^*)$,
\[ \ncut(S^*) \leq 2\,\cond(S^*) \leq 2 \cdot 3 \,\cond(T)  \leq 6 \,\ncut(T). \]
On the other hand, assumption~\eqref{assumption2} is satisfied with $\beta = 1/2$, so \cref{thm:ncut} guarantees that $\ncut(S^*) \leq 2 \ncut(T)$.
This is significantly tighter than combining the bound from \cref{thm:cond} and the relationship between normalized cut and conductance.
This result demonstrates that although HLC is presented as a localized variant of conductance, there are also situations in which we can obtain \emph{even better} set recovery guarantees in terms of normalized cut than conductance.
To summarize, our approach returns meaningful results in terms of more than just one clustering objective.

%% file: v2-experiments.tex

\section{Experiments}\label{sec:experiments}
We call running Alg.~\ref{alg:1} with Alg.~\ref{alg:local} as a subroutine \emph{HyperLocal}, 
since it operates on hypergraphs with a strongly-local runtime. 
Next, we apply HyperLocal to identify clusters of question topics on Stack Overflow and product categories within Amazon reviews.

\subsection{Algorithms and Implementation Details}
We implement HyperLocal in Julia, using a push-relabel implementation of the maximum $s$-$t$ flow method for the underlying $s$-$t$ cut problems.
All experiments ran on a laptop with 8 GB of RAM and a 2.2 GHz Intel Core i7 processor. We provide code and datasets at
\url{https://github.com/nveldt/HypergraphFlowClustering}.

\xhdr{Neighborhood Baselines} If $R_s$ is a set of seed nodes, let $\mathcal{N}(R_s)$ be its one-hop neighborhood.
In the hypergraphs we consider, the one-hop neighborhood of a seed set is often quite large.
We design two baselines for returning a cluster nearby a set of seeds.
$\mathit{TopNeighbors}$ orders nodes in $\mathcal{N}(R_s)$ based on the number of hyperedges that each $v \in \mathcal{N}(R_s)$ shares with at least one node from $R_s$
and outputs the top $k$ such nodes.
Similarly, $\mathit{BestNeighbors}$ orders each node $v \in \mathcal{N}(R_s)$ by the \emph{fraction} of hyperedges incident to $v$ that are also incident to at least one node from $R_s$
and outputs the top $k$.
In our experiments, we choose $k$ to be equal to ground truth cluster sizes, which provides an additional advantage to these baselines.

\xhdr{Clique Expansion + FlowSeed Baselines} FlowSeed \cite{Veldt2019flow} is a flow-based method for solving localized conductance~\eqref{localcond} in graphs.
For one baseline, we convert an input hypergraph to a graph and then run FlowSeed with the same input set $R$ and parameter $\varepsilon$ as we use for HyperLocal.
We consider two types of expansion: replacing a hyperedge $e$ with an unweighted clique, and replacing a hyperedge with a clique where each edge has weight $1/|e|$.
These are representative of existing clique expansion techniques for (local) hypergraph clustering~\cite{li2018tail,Zhou2006learning,hao2017local}. While other local graph clustering methods exist, we focus on comparing against FlowSeed, as it has been shown to outperform other techniques in a number of settings~\cite{Veldt2019flow}. Comparing against FlowSeed also allows us to best highlight the difference between running a flow-based method designed specifically for hypergraphs, versus performing clique reduction and applying a related graph algorithm.
Finally, comparing against non-local hypergraph clustering methods (e.g.,~\cite{panli2017inhomogeneous,panli_submodular}) is infeasible, as these do not seek clusters near a specified input set, and are unable to run on hypergraphs with millions of nodes and hyperedges in a reasonable amount of time.


\subsection{Question Topics on Stack Overflow}
HyperLocal is able to identify clusters of questions associated with the same topic on Stack Overflow.
We represent each question as a node and construct hyperedges from the set of questions answered by a single user.
Tags indicate sets of questions on the same topic (e.g., ``julia'', ``netsuite'', ``common-lisp''),
which we use as ground truth cluster labels (many questions have multiple tags).
The hypergraph has 15,211,989 nodes and 1,103,243 edges, with a mean hyperedge size of 23.7.
The dataset has 56,502 tags.
We use the 45 tags that have between 2,000 and 10,000 questions and a hypergraph conductance score below 0.2 under the all-or-nothing penalty.
Thus, we focus on sets of tags that can reasonably be viewed as modestly-sized clusters in the dataset.

\begin{table}[t]
  \caption{Average runtime in seconds, precision, recall, and F1 scores across 45 target clusters that correspond to question topics on Stack Overflow.
    \emph{Top F1} is the number of times out of 45 that a method/set obtained the top F1 score (including ties).
    Runtimes for last three rows are negligible.
    UCE and WCE indicate \emph{unweighted} and \emph{weighted} clique expansions.
  }
	\label{tab:stack}
	\centering
	\begin{tabular}{llllll}
		\toprule
		\textbf{Method}  & \emph{runtime} & \emph{pr} & \emph{re} & \emph{f1}  &  \emph{Top F1}\\
		\midrule 
		HyperLocal & 25.0 & 0.69 & 0.47 & 0.53 & 29\\
		UCE + FlowSeed & 32.8 & 0.3 & 0.58 & 0.4 & 1\\
		WCE + FlowSeed & 32.9 & 0.3 & 0.58 & 0.4 & 1\\
		BestNeighbors & -- & 0.49 & 0.49 & 0.49 & 11\\
		TopNeighbors & -- & 0.45 & 0.45 & 0.45 & 6\\
		$R$ & -- & 0.3 & 0.6 & 0.4 & 1\\
		\bottomrule
	\end{tabular}
\end{table} 

\begin{figure}[t]
	\centering
	\includegraphics[width=.95\linewidth]{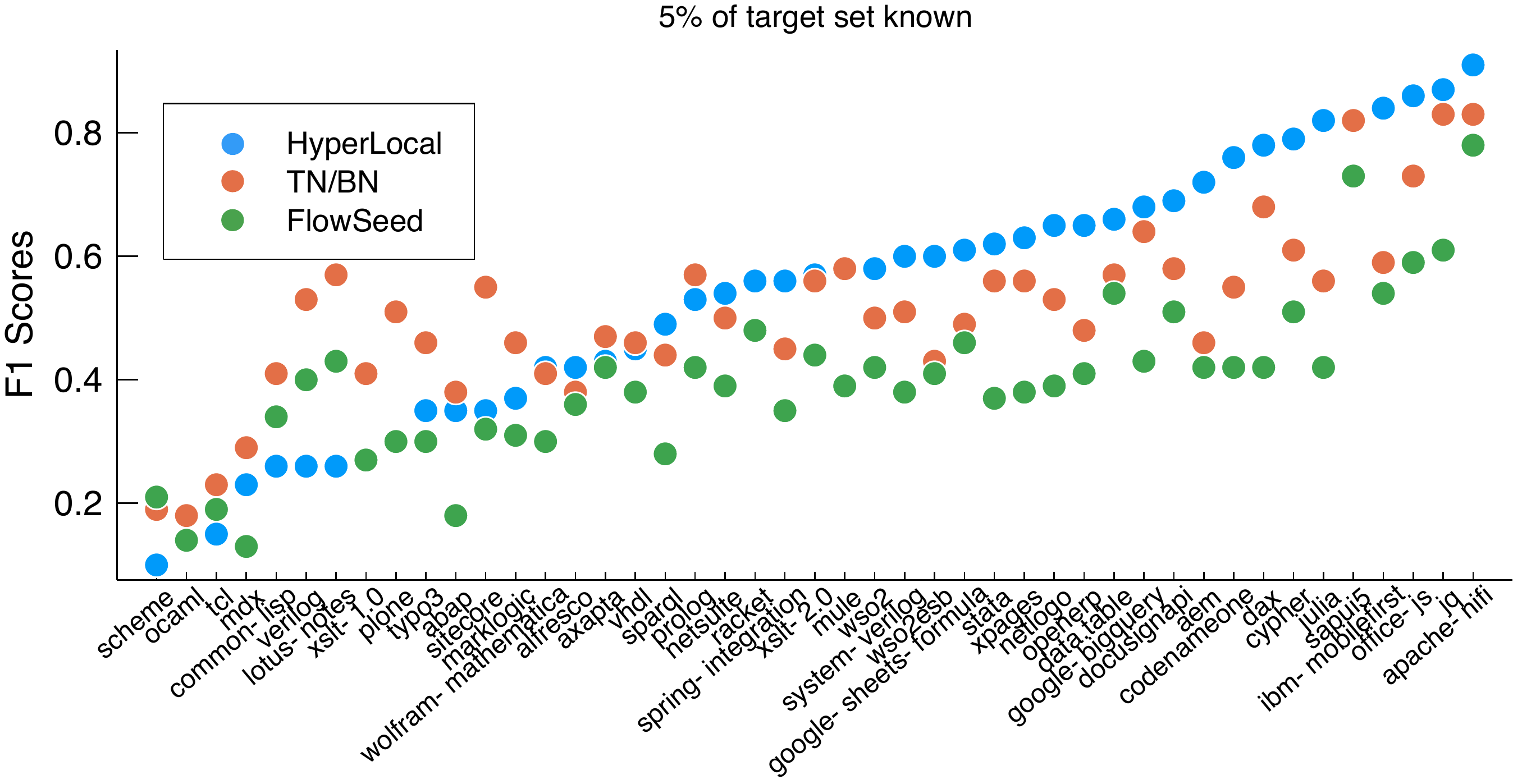}
	\includegraphics[width=.95\linewidth]{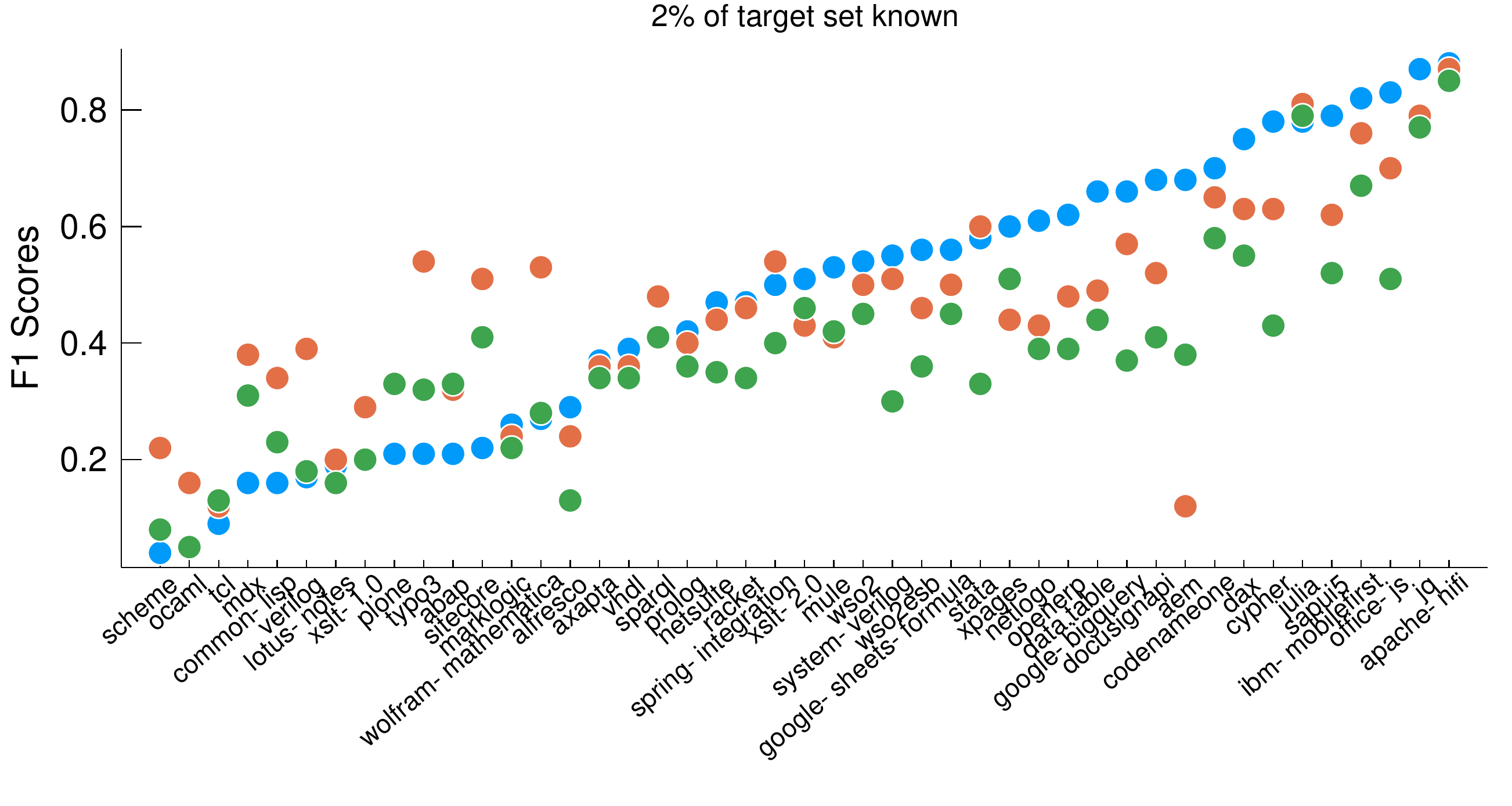}
	\caption{F1 scores for each of 45 clusters in the Stackoverflow hypergraph.
          For FlowSeed, we show the best result from the weighted or unweighted clique expansion;
          we also show the best of BN/TN.
          (Top) When 5\% of the target set is used as a seed set, HyperLocal has the highest mean F1 score, and outperforms all methods in all cases where at least one method has an F1 score above 0.6. (Bottom) Results are very similar when only 2\% of the target set is used as seeds.
        }
	\label{fig:stack}
\end{figure}

\xhdr{Experimental Setup}
To get an input set $R$ for Alg.~\ref{alg:1}, we assume that we are given a small subset of seed nodes $R_s$ from the target cluster and a rough idea of the cluster's size.
HyperLocal is designed to find good clusters by refining a moderately-sized reference set, so we use BestNeighbors to grow $R_s$ into an initial reference set $R \supseteq R_s$ that is refined by HyperLocal.
We ensure HyperLocal finds a cluster that strictly contains $R_s$ by adding infinite-weight edges from 
$R_s$ to the source node in the underlying minimum $s$-$t$ cut problems,
following prior approaches for local graph clustering~\cite{Veldt2019flow}.

For each target cluster $T$, we randomly select 5\% of $T$ as a seed set $R_s$, and use BestNeighbors to grow $R_s$ by an additional $2|T|$ nodes.
This produces a reference set input $R$ for HyperLocal.
We set $\varepsilon = 1.0$ and use the \deltasplit{} penalty with $\delta = 5000$.
This large threshold tends to produce good results (see \cref{sec:varying} for details),
and the threshold is meaningful as some hyperedges contain tens of thousands of nodes.

We run TopNeighbors and BestNeighbors, outputting the top $|T|$ nodes in the ordering defined by each.
Therefore, we give these methods an advantage by assuming that they know the exact size of the target cluster.
We also run FlowSeed on unweighted and weighted clique expansions, using the same parameters as HyperLocal.
In order to use clique expansion without densifying the graph too much and running into memory issues, we first discard hyperedges with 50 or more nodes (around 8\% of all hyperedges). Finally, we ran the full experimental procedure a second time, but starting with a random subset of 2\% of each target cluster, rather than 5\%.

\xhdr{Results}
\Cref{tab:stack} reports the performance of each method, averaged across all 45 clusters, when 5\% of each target set is known. HyperLocal has the highest average F1 score, and obtains the best F1 scores on many more clusters compared to other methods.
\Cref{fig:stack} visualizes F1 scores for individual clusters, both when 5\% and 2\% of the target cluster is known. Even when only 2\% of each target cluster is known, we see nearly identical results.

\subsection{Detecting Amazon Product Categories}
Next, we use HyperLocal to quickly detect clusters of retail products with the same category (e.g., ``Appliances'', ``Software'') from a large hypergraph constructed from Amazon product review data~\cite{ni-etal-2019-justifying}.
We construct hyperedges as sets of products (nodes) reviewed by the same person.
The hypergraph has 2,268,264 nodes and 4,285,363 hyperedges, with a mean hyperedge size just over 17.
We use product category labels as ground truth cluster identities and consider
the 9 smallest clusters, each of which represents only a very small fraction of
nodes in the hypergraph (\cref{tab:amazon}).

\xhdr{Experimental Setup}
We use $\varepsilon = 1.0$ and the standard all-or-nothing cut penalty, i.e.,
$\delta = 1$ for the \deltasplit{} penalty. Unlike the Stack Overflow dataset,
this smaller $\delta$ tends to work well (see \cref{sec:varying}).  For the six
smallest clusters (under 200 nodes), we use $|R_s| = 10$ random seed nodes and
use BestNeighbors to grow an initial cluster $R$ with $R_S$ plus 200 additional
nodes for HyperLocal to refine.  For the two clusters closer to 1000 nodes, we
use 50 seed nodes, which we grow by another 2,000 nodes using BestNeighbors.  The
largest two clusters have around 5,000 nodes. For these, we extract a random
subset of 200 nodes and use BestNeighbors to add 10,000 neighbors to form $R$. 

\xhdr{Results}
We compare HyperLocal to BestNeighbors and TopNeighbors in terms of F1 score.
(We also attempted to run FlowSeed, but were unable to perform a clique expansion on the hypergraph due to memory constraints, as the expanded graph becomes too dense even after removing all hyperedges with 50 nodes or more.)
\Cref{tab:amazon} reports F1 detection scores, averaged across 5 different trials with different random seed sets.
In all cases, HyperLocal substantially improves upon the baselines.

To test robustness, we ran numerous additional experiments on the smallest five clusters while varying the locality parameter $\varepsilon \in \{10^{-3}, 10^{-2}, 10^{-1}, 1, 10\}$ and reference set size $|R| \in \{200, 300, 500\}$.
In all cases, we obtained results similar to those in \cref{tab:amazon}.
Regarding runtime, HyperLocal takes between a few seconds and a few minutes, depending on the target cluster size.
This is remarkably fast considering that the method is repeatedly finding minimum $s$-$t$ cuts in a hypergraph with millions of nodes and hyperedges,
where the mean hyperedge size is above 17.

\begin{table}[t]
	\caption{The first column is the size of each Amazon product-category cluster $T$; 
	the second is HyperLocal (HL) runtime in seconds. 
	The remaining columns are F1 scores for HL, BestNeighbors (BN), TopNeighbors (TN), and the reference set $R$. 
	Except in one case, HyperLocal significantly improves on the F1 score of $R$ and always outperforms BN and TN.
	}
	\label{tab:amazon}
	\centering
	\begin{tabular}{lllllll}
		\toprule
		\textbf{Cluster} &$|T|$& run &  HL & BN & TN & $R$  \\
		\midrule 
		Amazon Fashion & 31 & 3.5& 0.83 & 0.77 & 0.6  & 0.67 \\
		All Beauty & 85 & 30.8& 0.69 & 0.6 & 0.28  & 0.58 \\
		Appliances & 48 & 9.8& 0.82 & 0.73 & 0.56  & 0.68 \\
		Gift Cards & 148 & 6.5& 0.86 & 0.75 & 0.71  & 0.63 \\
		Magazine Subscriptions & 157 & 14.5& 0.87 & 0.72 & 0.56  & 0.76 \\
		Luxury Beauty & 1581 & 261& 0.33 & 0.31 & 0.17  & 0.41 \\
		Software & 802 & 341& 0.74 & 0.52 & 0.24  & 0.42 \\
		Industrial \& Scientific & 5334 & 503 & 0.55 & 0.49 & 0.15  & 0.35 \\
		Prime Pantry & 4970 & 406 & 0.96 & 0.73 & 0.36  & 0.23 \\
		\bottomrule
	\end{tabular}
\end{table} 

\subsection{Varying Splitting Functions}\label{sec:varying}
 \begin{figure}[t]
	\centering
	\subfloat[Stackoverflow Question Topics \label{fig:stackdelta}] 
	{\includegraphics[width=.475\linewidth]{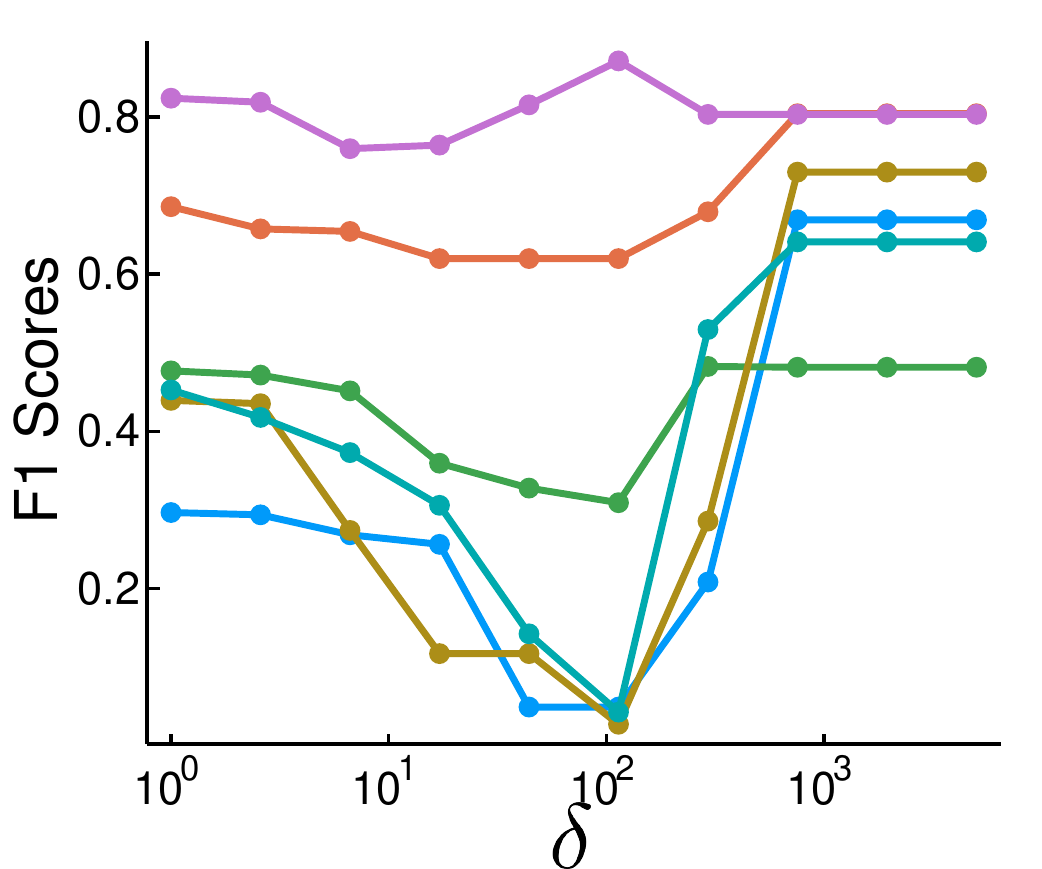}}
	\subfloat[Amazon Product Categories \label{fig:amazondelta}] 
	{\includegraphics[width=.475\linewidth]{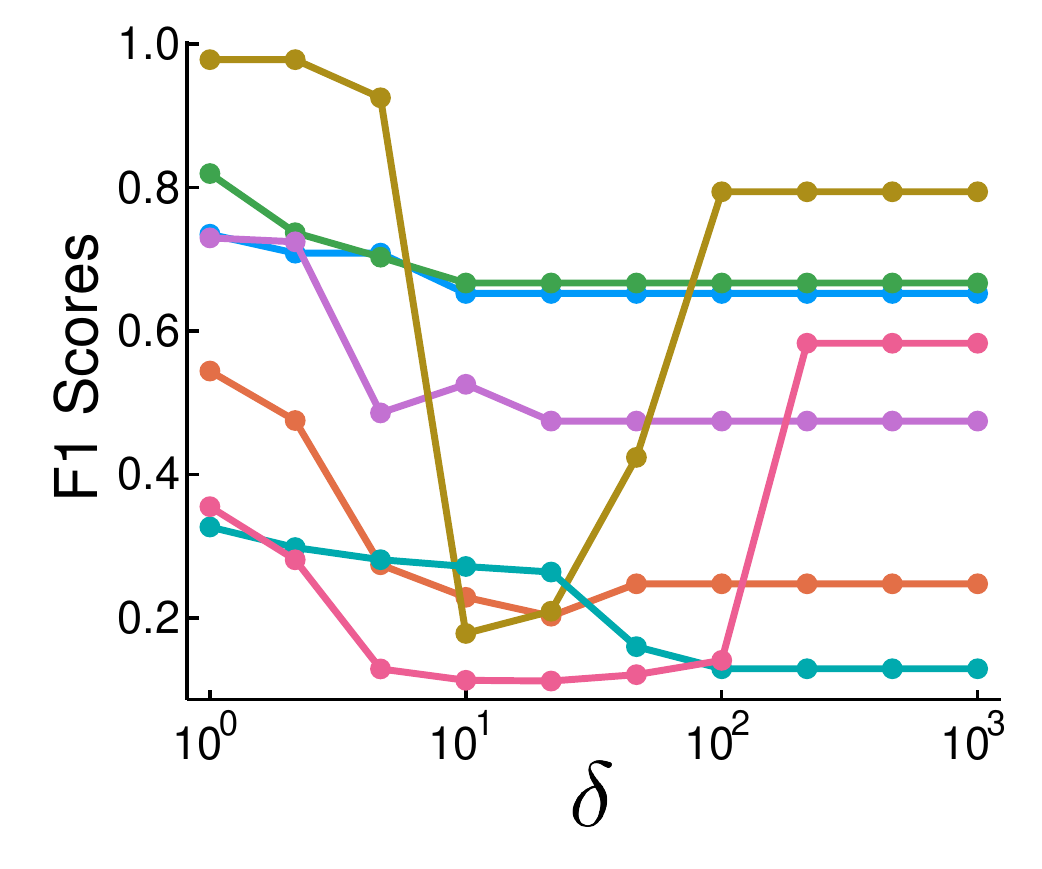}}\hfill
	\caption{F1 scores for HyperLocal for varying $\delta$.
          Each curve corresponds to a different cluster that HyperLocal is trying to detect from a seed set.
          For six Stack Overflow clusters, nearly all curves are maximized when $\delta > 1000$.
          For the seven smallest Amazon clusters (left), $\delta = 1.0$ leads to the highest F1 scores except in one case. }
	\label{fig:varydelta}
\end{figure}
In a survey on flow-based local \emph{graph} clustering, Fountoulakis et al.~\cite{fountoulakis2020flowbased} provide numerous guidelines and examples for choosing input sets $R$ and setting the resolution parameter $\varepsilon$. These guidelines apply in the same fashion to our hypergraph generalization. 


We additionally consider how different hypergraph splitting functions affect the output solution. We ran HyperLocal with the \deltasplit{} splitting function for varying $\delta$ for a handful of clusters from both datasets, measuring the target cluster recovery F1 score (\cref{fig:varydelta}). We tested both integer and non-integer values of $\delta$, and found no meaningful difference in runtime.
 For the Stack Overflow hypergraph, HyperLocal's performance plummets for $\delta$ near 100
 and performance is maximized for very large $\delta$ (\cref{fig:stackdelta}). 
 In contrast, for the Amazon hypergraph, the maximum F1 for each cluster is almost always obtained when $\delta = 1.0$ (\cref{fig:amazondelta}).
 This suggests that if one has access to a subset of ground truth clusters in a hypergraph,
 running a similar set of experiments provides a simple and effective strategy for choosing a value of $\delta$ to use when 
 searching for other, similar clusters.

%% file: v1-conclusion.tex

\section{Conclusion}
We have presented the first strongly-local method for localized hypergraph clustering, 
which comes with theoretical guarantees in terms of output quality and runtime.
One attractive property of our method is that it can leverage recent results on generalized hypergraph s-t cut problems in order to minimize localized ratio cut objectives with general hypergedge splitting functions.
Here we have considered one new parametric family of splitting functions (\deltasplit{}).
Finding other families of splittings functions could be used to detect other types of meaningful clustering structure in large networks offers opportunity for future research.

%% file: sec-appendix-proofs.tex

\section{Proofs for Theorems}

\subsection{Proof of Theorem~\ref{thm:itbound}}
\begin{proof}
	We want to show that $\cut_\mathcal{H}(S)$ strictly decreases in each pass of the while loop (except the last) in Algorithm~\ref{alg:1}. 
	To simplify notation, we drop the terms $\mathcal{H}$, $R$ and $\varepsilon$ from function subscripts, since these are fixed and clear from context. 
	Consider a pair of consecutive $s$-$t$ cut solutions where there is a strict improvement in HLC score. Starting with some value $\alpha_{j-1}$, let $S_j = \argmin \textbf{H-st-cut}_{\alpha_{j-1}}(S)$, and $\alpha_j = \hlc(S)$, with $\alpha_j < \alpha_{j-1}$. Let $S_{j+1} = \textbf{H-st-cut}_{\alpha_{j}}(S)$ be the set obtained in the next pass through the while loop, with $\alpha_{j+1} =  \hlc(S_{j+1})$. Since we are assuming that HLC improves in both of these steps, we have $\alpha_{j+1} < \alpha_j < \alpha_{j-1}$. Now observe that
	\begin{align*}
	&\textbf{H-st-cut}_{\alpha_{j-1}}(S_j)\\
	&=  \cut(S_j) + \alpha_{j-1} \vol(\bar{S}_j \cap R) + \alpha_{j-1} \varepsilon \vol(S_j \cap \bar{R}) \\
	&= \cut(S_j) - \alpha_{j-1} \vol(S_j \cap R) + \alpha_{j-1} \varepsilon \vol(S_j \cap \bar{R})+ \alpha_{j-1} \vol(R)\\
	  &= \cut(S_j) - \alpha_{j-1}\Omega(S_j) + \alpha_{j-1}\vol(R) \\
          &= \alpha_{j-1}\vol(R) + \Omega(S_j)(\hlc(S_j) - \alpha_{j-1}) \\
          &= \alpha_{j-1}\vol(R) + \Omega(S_j)(\alpha_j - \alpha_{j-1}). 
	\end{align*}
	The same essential steps show that
	\[
	\textbf{H-st-cut}_{\alpha_{j-1}}(S_{j+1}) = \alpha_{j-1}\vol(R) + \Omega(S_{j+1}) (\alpha_{j+1} - \alpha_{j-1}).
	\]
	We know that $\textbf{H-st-cut}_{\alpha_{j-1}}(S_j) \leq \textbf{H-st-cut}_{\alpha_{j-1}}(S_{j+1})$, since $S_j$ is the optimal $s$-$t$ cut solution for parameter $\alpha_{j-1}$. This implies that
	\[
	\Omega(S_j) (\alpha_j - \alpha_{j-1}) \leq \Omega(S_{j+1})(\alpha_{j+1} - \alpha_{j-1}),
	\]
	which in turn means that $\Omega(S_{j+1}) < \Omega(S_{j})$, since $(\alpha_{j+1} - \alpha_{j-1}) < (\alpha_{j} - \alpha_{j-1}) < 0$. Finally, because the HLC score and its denominator $\Omega$ decrease when going from $S_j$ to $S_{j+1}$, it must also be the case that $\cut(S_{j+1}) < \cut(S_j)$. Thus, until the last step of Algorithm~\ref{alg:1}, the cut function is strictly decreasing.
\end{proof}

\subsection{Proof of Theorem~\ref{thm:converges}}
\begin{proof}
  For any $v \in S$, the set of nodes and edges that are adjacent to $v$ in $\mathcal{L}$ is exactly the same as the set of nodes and edges that are adjacent to $v$ in $\mathcal{H}_\alpha$.
  The reason is that when the algorithm terminates ($S_{L} = S_{L-1}$ for some $L$), any $v$ in the output has either been \emph{explored} (or is in $R$) and had its neighbors added to $\mathcal{L}$.
  Therefore, $\textbf{L-st-cut}_\alpha(S) = \textbf{H-st-cut}_\alpha(S)$.
  Since the function $\textbf{L-st-cut}_\alpha$ is a lower bound on $\textbf{H-st-cut}_\alpha$ in general,
	\begin{align*}
	\textbf{L-st-cut}_\alpha(S) &= \textbf{H-st-cut}_\alpha(S) 
	\geq \min_A \textbf{H-st-cut}_\alpha(A) \\
	&\geq \min_A \textbf{L-st-cut}_\alpha(A) = \textbf{L-st-cut}_\alpha(S)\,.
	\end{align*}
	Thus, $S = \argmin_A \textbf{H-st-cut}_\alpha(A)$.
\end{proof}

\subsection{Proof of Theorem~\ref{thm:bounds}}
\begin{proof}
  While Algorithm~\ref{alg:local} does not rely on explicitly applying graph reduction techniques nor computing maximum $s$-$t$ flows, our proof will rely on the existence of both, as well as on a basic understanding of the max-flow min-cut theorem.
	
	\xhdr{Implicit graph $s$-$t$ cuts}
	Let $L^i = (V_i \cup \{s,t\}, E_i \cup E_i^{st})$ be the local hypergraph over which we solve a minimum $s$-$t$ cut in the $i$th iteration of the algorithm, where $V_i \subseteq V$, $E_i \subseteq E$, and $E_i^{st} \subseteq E^{st}$. Let $S_i$ be the minimum $s$-$t$ cut set for $L_i$, and let $N_i \subseteq S_i$ be the set of nodes that are \emph{explored} in the $i$th iteration, i.e., nodes whose terminal edges are cut for the first time in iteration $i$.
	Since we assume that the hypergraph cut function is graph reducible, for each $L_i$ there exists a graph $G_i$ with the same set of nodes $V_i \cup\{s,t\}$,
        plus potentially other auxiliary nodes, such that the minimum $s$-$t$ cut value in $G_i$ is the minimum $s$-$t$ cut value in $L_i$.
        Formally, let $S'_i$ be the minimum $s$-$t$ cut set in $G_i$ (excluding $s$ itself), so that $S_i = S_i' \cap V_i$. In other words, if we exclude auxiliary nodes, the minimum $s$-$t$ cut set in $G_i$ is the minimum $s$-$t$ cut set in $L_i$.
	
	\xhdr{Bounding set sizes}
	Let $S$ be any subset of vertices such that $S$ contains no isolated nodes.
        Let $E(S,S)$ denote the set of hyperedges that are completely contained inside $S$.
        We have the following bounds:
	\begin{align}
	\label{partialS}
	|\partial S| &\leq \vol_\mathcal{H}(S) \\
	\label{S}
	|S| & \leq \vol_\mathcal{H}(S) \\
	\label{ES}
	E(S,S) &\leq \frac{\vol_\mathcal{H}(S)}{2}\,.
	\end{align}
        These bounds use the theorem assumptions on the splitting function; the fact that the minimum
        weight is one means that the volume of a node is equal to its degree.
	The first bound is tight whenever every edge that is cut by $S$ contains exactly one node from $S$.
        The second is tight when every node in $S$ has degree one.
        Bound~\eqref{ES} follows from the fact that every hyperedge is of size at least 2, and therefore each hyperedge that is completely contained $S$ is made up of at least two nodes from $S$.
        For $k$-uniform hypergraphs, $E(S,S) \leq \frac{\vol_\mathcal{H}(S)}{k}$, though we use the bound~\eqref{ES} so that we can apply our results more generally. 

        Assume \cref{alg:local} terminates after iteration $t$, so that $L_t$ is the largest local hypergraph formed.
        Let $P$ denote the set of nodes that were \emph{explored} at some point during the algorithm:
	\begin{equation}
	\label{p}
	P = \bigcup_{i=1}^t N_i\,.
	\end{equation}
	We will prove later that the volume of $P$ can be bounded as follows:
	\begin{equation}
	\label{pvol} 
	\vol_\mathcal{H}(P) \leq \frac{\vol_\mathcal{H}(R)}{\varepsilon}.
	\end{equation}
	For now, we assume this to be true and use it to prove the bounds given in the statement of the theorem.
	
	Let $Q$ denote the set of nodes in $L_t$ that were never {explored}.  The size of this set can be bounded as follows:
	\begin{equation}
	|Q| \leq (k-1)(|\partial R| + |\partial P|).
	\end{equation}
	This bound will often be quite loose in practice. However, in theory it is possible for a hyperedge in $L_t$ to contain only one node from $R\cup P$, and $(k-1)$ nodes from the set $Q$. We bound the total number of nodes in $L_t$ with help from Eqs.~\eqref{partialS},~\eqref{S}, and~\eqref{p}:
	\begin{align*}
	|V_t| &= |R| + |P| + |Q| \leq \vol_\mathcal{H}(R) + \vol_\mathcal{H}(P) + (k-1)(|\partial R| + |\partial P|)\\
	&\leq \vol_\mathcal{H}(R)(1 + 1/\varepsilon) + (k-1)\big(\vol_\mathcal{H}(R) + \vol_\mathcal{H}(P)  \big) \\
	&\leq k \vol_\mathcal{H}(R)(1 + 1/\varepsilon).
	\end{align*}	
	Note also that $|V_t|$ is the exact number of terminal edges in $L_t$, so we also have a bound on the number of terminal edges.
	
	We can bound the number of hyperedges $|E_t|$ in $L_t$ above by $|E(R,R)| + |E(P,P)| + |\partial R| + |\partial P|.$ Note that any hyperedge that includes a node from $Q$ is accounted for by the terms $|\partial R|$ and $|\partial P|$. We again use bounds~\eqref{partialS},~\eqref{S},~\eqref{ES} and~\eqref{p} to bound the number of hyperedges in terms of $\vol(R)$:
	\begin{align*}
	|E_t| &\leq |E(R,R)| + |E(P,P)| + |\partial R| + |\partial P|\\ &\leq \frac{3}{2}\big(\vol_\mathcal{H}(R) + \vol_\mathcal{H}(P)  \big)
	\leq \frac{3}{2}\Big(1 + \frac{1}{\varepsilon}\Big) \vol_\mathcal{H}(R).
	\end{align*}
	
	The last step of the proof is to show the volume bound on $P$ in~\eqref{p}, which we do by proving the existence of a maximum $s$-$t$ flow on a graph reduction of $L_t$ with certain properties.
	
	\xhdr{Bounding $P$ with an implicit maximum flow argument}
	The max-flow min-cut theorem states that the value of the minimum $s$-$t$ cut in a graph $G$ is equal to the maximum $s$-$t$ flow value in $G$.  An edge is \emph{saturated} if the value of the flow on an edge equals the weight of the edge, which always upper bounds the flow. Given a set of edges $C$ defining a minimum $s$-$t$ cut $G$, \emph{any} maximum $s$-$t$ flow $F$ in $G$ must saturate \emph{all} edges in $C$. If any edge in $C$ were not saturated by $F$, then the flow value would be strictly less than the cut value, contradicting the optimality of either $F$ or $C$. We will use this understanding of the max-flow min-cut theorem to prove the existence of a flow that saturates all edges of \emph{explored} nodes in the local hypergraph.
	
	The set $N_1$ is made up of all nodes whose edge to the sink is cut in $G_1$ when we compute a minimum $s$-$t$ cut. 
	We know that even if we do not compute it explicitly, there exists some flow $F_1$ in $G_1$ that saturates all edges between $N_1$ and the sink $t$.
        In the next iteration, $N_2$ is the set of nodes whose edges to the sink are cut for the first time.
        One way to compute a maximum $s$-$t$ flow $F_2$ in $G_2$ is to start with $F_1$, the maximum $s$-$t$ flow in $G_1$, and then find new augmenting flow paths until no more flow can be routed from $s$ to $t$. We can assume without loss of generality that the terminal edges of $N_1$ remain saturated by $F_2$, since there can be no net gain from reversing the flow on a saturated edge to the sink. As a result, the flow $F_2$ will saturate the terminal edges of $N_1$ \emph{as well as} all terminal edges of $N_2$, which are cut by the minimum $s$-$t$ cut in $G_2$. Continuing this process inductively, we note that in the $i$th iteration there exists some flow $F_i$ that saturates \emph{all} the terminal edges to the sink that have been cut by some $s$-$t$ cut in a previous iteration. In other words, there exists some flow $F_t$ in $G_t$ that saturates the terminal edge of every node in $P = \cup_{i=1}^t N_i$. 
	Recall that
	\begin{equation}
	\text{Weight of terminal edges of $P$} = \sum_{v \in P} \alpha \varepsilon d_v = \alpha \varepsilon \vol_\mathcal{H}(P)\,.
	\end{equation}
	Finally, observe that the minimum $s$-$t$ cut score in $G_t$ is bounded above by $\alpha \vol_\mathcal{H}(R)$, since this is the weight of edges adjacent to the source. This provides an upper bound on the weight of $P$'s terminal edges, implying the desired bound on the volume of $P$:
	\begin{equation}
	\alpha \varepsilon \vol_\mathcal{H} (P) \leq \alpha \vol_\mathcal{H}(R) \implies \vol_\mathcal{H}(P) \leq \frac{\vol_\mathcal{H}(R)}{\varepsilon}.
	\end{equation}
\end{proof}

\subsection{Proof of Theorem~\ref{thm:ncut}}
\begin{proof}
	Assume throughout that we deal only with sets $S$ satisfying $\Omega_{R, \varepsilon}(S) > 0$, to avoid trivial cases. 
	We begin by defining
	\begin{equation}
	\label{g1}
	g(S) = \vol(\bar{R}) \cdot\vol(S \cap R) - \vol(R) \cdot\vol(S \cap \bar{R}) .
	\end{equation}
	Dividing every term in $g(S)$ by $\vol(\bar{R})$ gives
	\begin{equation}
	{g(S)}/{\vol(\bar{R})} = \vol(S \cap R) - \varepsilon_0 \vol(S \cap \bar{R}).
	\end{equation}
	This allows us to re-write the HLC objective as
	\begin{equation}
	\label{hlc_2}
	\hlc = \vbr \cdot\frac{\cut(S)}{g(S)- \mu\vol(\bar{R})  \cdot\vol(S \cap \bar{R})}.
	\end{equation}
	Applying a few steps of algebra produces another useful characterization of the function $g$:
	\begin{equation}
	\label{g2}
	\begin{split}
	g(S) &=\vol(\bar{R}) \cdot\vol(S \cap R) - \vol(R) \cdot\vol(S \cap \bar{R}) \\
	&= \vol(V) \cdot\vol(S \cap R) - \vol(R) \cdot\vol(S) \\
	&= \big(\vol(V) - \vol(S)\big) \cdot\vol(S \cap R) - \vol(S)\cdot \vol( \bar{S} \cap R) \\
	&=\vol(\bar{S})\cdot \vol(S \cap R) - \vol(S) \cdot\vol( \bar{S} \cap R) \,.
	\end{split}
	\end{equation}
	This characterization of $g$ allows us to see that for any $S \subset V$:
	\begin{equation}
	\label{first bound}
	\begin{split}
	\vol(\bar{S}) \cdot \vol(S)
	&\geq \vol(\bar{S}) \cdot\vol(S \cap R)  \geq g(S) \\
	&\geq g(S) - \mu \vbr \cdot\vol(S \cap \bar{R}).
	\end{split}
	\end{equation}
	We use this to upper bound $\ncut(S)$ in terms of $\hlc(S)$:
	\begin{align}
	\label{lower}
	\frac{\vbr}{\vol(V)} \cdot \ncut(S)  &= \vbr \cdot \frac{\cut(S)}{\vol(S) \cdot\vol(\bar{S})} \leq \hlc(S).
	\end{align}
	Next we need to prove a \emph{lower} bound on the normalized cut score of $T$. We again use the characterization of $g$ given in~\eqref{g2}, this time in conjunction with property~\eqref{assumption2}, satisfied by $T$, to see that
	\begin{equation}
	\label{gt}
	g(T) \geq \beta \vol(T) \cdot\vol(\bar{T}).
	\end{equation}
	Property~\eqref{assumption2} also implies that
	\begin{align*}
	1 - \frac{\vol(\bar{T}\cap R)}{\vol(T)} = \frac{\vol(T\cap R)}{\vol(T)} \geq \beta \implies \vol(\bar{T} \cap R) \leq (1-\beta) \vol(T).
	\end{align*}
	Combining this with $\vol(\bar{R}) \leq \vol(V) \leq 2\vol(\bar{T})$ produces
	\begin{equation}
	\label{helpful2}
	\vol(\bar{R}) \cdot \vol(T \cap \bar{R}) \leq 2(1-\beta) \vol(T) \cdot\vol(\bar{T}).
	\end{equation}
	Inequalities~\eqref{gt} and~\eqref{helpful2} together imply that
	\begin{equation}
	\label{main_inequalities}
	\begin{split}
	\vbr \cdot \Omega_{R,\varepsilon}(T)
	&=g(T) - \mu \vbr \cdot\vol(T \cap \bar{R})\\
	&\geq \big(\beta - 2\mu(1-\beta) \big)\vol(T) \cdot\vol(\bar{T}).
	\end{split}
	\end{equation}
	Finally, we put together the bound~\eqref{lower}, the characterization of the HLC objective given in~\eqref{hlc_2}, and inequality~\eqref{main_inequalities}, to see that
	\begin{align*}
	  \frac{\vbr}{\vol(V)} \cdot \ncut(S^*)
	&\leq \hlc(S^*) \leq \hlc(T) \\
	& = \vbr \cdot \frac{\cut(T)}{g(T) - \mu \vbr \cdot\vol(T \cap \bar{R})} \\
	& \leq \vbr \cdot \frac{1}{\beta - 2\mu(1-\beta)} \cdot  \frac{\cut(T)}{\vol(T) \cdot\vol(\bar{T})} \\
	&= \frac{\vbr}{\vol(V)}\cdot \frac{\ncut(T)}{\beta - 2\mu(1-\beta)}\,.
	\end{align*}
	Dividing through by $\frac{\vbr}{\vol(V)}$ yields the desired bound on normalized cut.
\end{proof}